\documentclass[1p,times,sort&compress]{elsarticle}
\usepackage{amsfonts,amsmath,fixltx2e,verbatim}
\usepackage{amsthm}
\DeclareMathOperator{\Ff}{Fact}
\DeclareMathOperator{\Pre}{Pref}

\DeclareMathOperator{\card}{card}
\newcommand{\vt}{\vartheta}
\newcommand{\Nn}{\mathbb N}
\newcommand{\Int}{\mathbb Z}
\newcommand{\PAL}{\mathit{PAL}}
\newcommand{\PR}{\mathit{PRIM}}
\newtheorem{thm}{Theorem}[section]
\newtheorem{prop}[thm]{Proposition}
\newtheorem{cor}[thm]{Corollary}
\newtheorem{lemma}[thm]{Lemma}
\theoremstyle{remark}
\newtheorem{example}{Example}[section]

\theoremstyle{definition}
\newtheorem{remark}[thm]{Remark}
\begin{document}
\begin{frontmatter}
\title{A generalized palindromization map in free monoids}
\author[dma]{Aldo de Luca\corref{cor1}}
\ead{aldo.deluca@unina.it}
\author[dsf]{Alessandro De Luca}
\ead{alessandro.deluca@unina.it}
\cortext[cor1]{Corresponding author}
\address[dma]{Dipartimento di Matematica e Applicazioni ``R.~Caccioppoli''\\
Universit\`a degli Studi di Napoli Federico II\\
via Cintia, Monte S.~Angelo ---
I-80126 Napoli, Italy}
\address[dsf]{Dipartimento di Scienze Fisiche\\
Universit\`a degli Studi di Napoli Federico II\\
via Cintia, Monte S.~Angelo ---
I-80126 Napoli, Italy}

\begin{abstract}
The palindromization map $\psi$ in a free monoid $A^*$ was introduced in 1997 by the first author in the case of a binary alphabet $A$, and later extended  by other authors to arbitrary alphabets. Acting on infinite words, $\psi$ generates the class of standard episturmian words, including standard Arnoux-Rauzy words. 
In this paper we generalize the palindromization map, starting with a given code $X$ over $A$. The new map $\psi_X$  maps $X^*$ to the set $PAL$ of palindromes of $A^*$. In this way some properties of $\psi$ are lost and some are saved in a weak form.
 When $X$ has a finite deciphering delay one can extend $\psi_X$ to 
$X^{\omega}$, generating a class of infinite words much wider than standard episturmian words. For a finite and maximal code $X$ over $A$, we give a suitable generalization of standard Arnoux-Rauzy words, called $X$-AR words. We prove that any $X$-AR word is a morphic image of a standard Arnoux-Rauzy word and we determine some suitable linear lower and upper bounds to its factor complexity.  
 
  For any code $X$ we say that $\psi_X$ is conservative when
$\psi_X(X^{*})\subseteq X^{*}$. We study conservative maps $\psi_X$ and conditions on $X$ assuring that $\psi_X$ is conservative. We also investigate the special case of morphic-conservative maps $\psi_{X}$, i.e., maps such that $\varphi\circ \psi = \psi_X\circ \varphi$ for an injective morphism $\varphi$. 
 Finally, we generalize $\psi_X$ by replacing palindromic closure  with $\vartheta$-palindromic closure, where $\vartheta$ is any involutory antimorphism of $A^*$. This yields an extension of the class of
$\vartheta$-standard words introduced by the authors in 2006.
 \end{abstract}

\end{frontmatter}

\section{Introduction}

  A simple method  of constructing all standard Sturmian words  was introduced by the first author in \cite{deluca}. It is based on an operator definable in any free monoid $A^*$ and called right palindromic closure,
which maps  each word $w\in A^*$ into the shortest palindrome of $A^*$ having $w$ as a prefix. Any given word $v\in A^*$  can suitably `direct' subsequent  iterations of  the preceding operator  according to the sequence of letters  in $v$, as follows: at each step, one concatenates the next letter of $v$ to the right of the already constructed palindrome and then takes the right palindromic closure. Thus, starting with any directive word $v$, one generates a palindrome $\psi(v)$. The map $\psi$,  called palindromization map, is injective; the word $v$ is called the directive word of $\psi(v)$. 

Since for any
$u,v\in A^*$, $\psi(uv)$ has $\psi(u)$ as a prefix, one can extend the map $\psi$ to right infinite words
$x\in A^{\omega}$ producing an infinite word $\psi(x)$. It has been proved in \cite{deluca} that   if each letter  of a binary alphabet $A$ occurs infinitely often in $x$, then one can generate all standard Sturmian words.

 The palindromization map $\psi$ has been extended to   infinite words over an arbitrary alphabet $A$ by X. Droubay, J. Justin, and G. Pirillo in  \cite{DJP}, where  the family  of  standard episturmian words over $A$
has been introduced.  In the case that each letter of $A$  occurs infinitely often in the directive word,  one obtains the class of  standard Arnoux-Rauzy words \cite{AR, Rauzy}. A standard Arnoux-Rauzy word over a binary alphabet is a  standard Sturmian word.

Some generalizations of the palindromization map have been given. In particular, in \cite{adlADL} a $\vartheta$-palindromization map, where $\vartheta$ is any involutory antimorphism of a free monoid, has been introduced. By acting with this operator
on any infinite word one obtains  a class of words larger than the class
of standard episturmian,  called  $\vartheta$-standard words; when $\vartheta$ is the reversal operator one obtains the class of standard episturmian words. Moreover, the palindromization map  has been recently extended to the case of the free group $F_2$  by C. Kassel and C. Reutenauer in \cite{KREU}.
A recent survey on palindromization map and its generalizations is in \cite{adl011}. 

In this paper we introduce  a natural generalization of the palindromization map which is considerably more powerful than the map $\psi$ since it allows to generate a class of infinite words much wider than standard episturmian words. The generalization is obtained by replacing the alphabet $A$ with a code $X$ over $A$ and then   `directing' the successive applications of the right-palindromic closure operator by a sequence of words of the code $X$. Since any non-empty element of $X^*$ can be uniquely factorized by the words of $X$, one can uniquely map any word of $X^*$ to a palindrome. In this way it is possible associate to every code $X$ over $A$ a generalized palindromization map denoted by $\psi_X$. If $X=A$ one reobtains the usual palindromization map.

  General properties of  the map $\psi_X$ are considered in Section \ref{sec:three}.  Some properties satisfied by  $\psi$ are lost and others are saved in a weak form.  In general $\psi_X$ is not injective; if $X$ is a prefix code, then  $\psi_X$ is injective. Moreover, for any code $X$,  $w\in X^*$, and $x\in X$ one has that $\psi_X(w)$ is a prefix of $\psi_X(wx)$.

In Section \ref{sec:four} the generalized palindromization map is extended to infinite words of $X^{\omega}$. In order to define a map $\psi_X: X^{\omega} \rightarrow A^{\omega}$ one needs that the code $X$ has a finite deciphering delay, i.e., any word of $X^{\omega}$ can be uniquely factorized in terms of the elements of $X$. 
For any $t\in X^{\omega}$ the word $s=\psi_X(t)$ is trivially closed under reversal, i.e., if $u$ is a factor of $s$, then so will be its reversal $u^{\sim}$.  If $X$ is a prefix code, the map $\psi_X: X^{\omega} \rightarrow A^{\omega}$ is injective. Moreover,
one can prove that if $X$ is a finite code having a finite deciphering delay, then for any $t\in X^{\omega}$ the word
$\psi_X(t)$ is uniformly recurrent. 
We show that one can generate all standard Sturmian words by  the palindromization map $\psi_X$
with  $X = A^2$. Furthermore, one can also construct the Thue-Morse word by using the generalized palindromization map relative to a suitable infinite code.

In Section \ref{sec:five} we consider the case of a map $\psi_X: X^{\omega} \rightarrow A^{\omega}$
in the hypothesis that $X$ is a maximal finite code. From a basic theorem of Sch\"utzenberger
 the code $X$ must have a deciphering delay equal to $0$, i.e.,  $X$ has to be a maximal prefix code. 
 Given $y=x_1\cdots x_i \cdots \in X^{\omega}$ with $x_i\in X$, $i\geq 1$, we say that the word $s=\psi_X(y)$  is a generalized Arnoux-Rauzy word relative to $X$, briefly $X$-AR word,
if for any word $x\in X$ there exist infinitely many integers $i$ such that  $x=x_i$.  If $X=A$ one obtains the usual definition of standard Arnoux-Rauzy word. 

Some  properties of the generalized Arnoux-Rauzy words are proved. In particular, any $X$-AR word $s$ is $\omega$-power free, i.e., any non-empty factor of $s$  has a power which is not a factor of $s$. We prove that the number $S_r(n)$ of right special factors of $s$ of length $n$ for a sufficiently large $n$ has the lower bound  given by the number of proper prefixes of $X$, i.e., $(\card(X)-1)$/$(d-1)$, where $d= \card(A)$. From this one obtains that for a sufficiently large $n$, the factor complexity $p_s(n)$ has the lower bound $(\card(X)-1)n +c$, with $c\in \Int $. Moreover,
we  prove that for all $n$, $p_s(n)$ has the linear upper bound $2\card(X)n +b$ with $b\in \Int $.
The proof of this latter result is based on a theorem which gives a suitable generalization of a formula of Justin \cite{J}. A further consequence of this theorem is that any $X$-AR word is a morphic image
of a standard Arnoux-Rauzy word on an alphabet of $\card(X)$ letters.
An interesting property showing that any $X$-AR word $s$ belongs to $ X^{\omega}$ is proved in Section \ref{sec:six}.

In Section \ref{sec:six} we consider a palindromization map $\psi_X$  satisfying the condition $\psi_X(X^*)\subseteq X^*$. We say that  $\psi_X$ is conservative. Some general properties of  conservative maps are studied and a sufficient condition on $X$ assuring that $\psi_X$ is conservative is given. 
A special case of conservative map is the following: let $\varphi: A^*\rightarrow B^*$ be an injective morphism such that $\varphi(A) = X$. The map $\psi_X$ is called morphic-conservative if for all $w\in A^*$, $\varphi(\psi(w)) = \psi_X(\varphi(w))$. We prove that if  $\psi_X$ is morphic-conservative, then
$X\subseteq PAL$, where $PAL$  is the set of palindromes, and $X$ has to be a bifix code. This implies that $\psi_X$ is injective. Moreover one has that $\psi_X$ is morphic-conservative if and only if $X\subseteq PAL$, $X$ is prefix, and $\psi_X$ is conservative.
Any morphic-conservative map $\psi_X$ can be extended to $X^{\omega}$ and the infinite words which are generated are images by an injective morphism of  epistandard words. 
An interesting generalization of conservative map  to the case of infinite words is the following: a map $\psi_X$, with $X$ a code having a finite deciphering delay,  is weakly conservative if for any $t\in X^{\omega}$, $\psi_X(t)\in X^{\omega}$. If $\psi_X$ is conservative, then it is trivially weakly conservative, whereas the converse is not in general true. We prove that if $X$ is a finite maximal code, then $\psi_X$ is weakly conservative.

In Section \ref{sec:seven} we give an extension of the generalized palindromization map $\psi_X$ by replacing the palindromic closure operator with the $\vartheta$-palindromic closure operator, where
$\vartheta$ is an arbitrary involutory antimorphism in $A^*$. In this way one can define a generalized
$\vartheta$-palindromization map $\psi_{\vartheta,X}: X^* \rightarrow PAL_{\vartheta}$, where $PAL_{\vartheta}$ is the set of fixed points of $\vartheta$ ($\vartheta$-palindromes). If $X$ is a code having a finite deciphering delay one can extend $\psi_{\vartheta,X}$ to $X^{\omega}$ obtaining
a class of infinite words larger than the $\vartheta$-standard words introduced in \cite{adlADL}.
We limit ourselves to proving a noteworthy theorem showing that  $\psi_{\vartheta}= \mu_{\vartheta}\circ \psi= \psi_{\vartheta, X}\circ \mu_{\vartheta}$ where $X= \mu_{\vartheta}(A)$ and $\mu_{\vartheta}$ is the injective morphism defined for any $a\in A$ as  $\mu_{\vartheta}(a)= a$ if $a=\vartheta(a)$ and
$\mu_{\vartheta}(a)= a\vartheta(a)$, otherwise.

\section{Notation and preliminaries}\label{sec:two}

Let $A$ be a non-empty finite set, or \emph{alphabet}. In the following, $A^*$
 will denote the \emph{free monoid} 
generated by $A$. The elements of $A$ are called \emph{letters} and those of
$A^*$ \emph{words}.  The identity element of $A^*$ is called \emph{empty word}
and it is denoted by $\varepsilon$. We shall set  $A^+= A^*\setminus \{\varepsilon\}$. A word $w\in A^+$ can be written uniquely
as a product of letters $w=a_1a_2\cdots a_n$, with $a_i\in A$,
$i=1,\ldots,n$. The integer $n$ is called the \emph{length} of $w$ and is
denoted by $|w|$. The length of $\varepsilon$ is conventionally 0.

Let $w\in A^*$. A word $v$ is a \emph{factor} of $w$ if there exist words $r$ and $s$ such that $w=rvs$. A factor $v$ of $w$ is \emph{proper} if $v\neq w$. If $r=\varepsilon$ (resp. $s=\varepsilon$),
 then $v$ is called a \emph{prefix} (resp.  \emph{suffix}) of $w$. If $v$ is a prefix (resp. suffix) of $w$, then $v^{-1}w$ (resp. $wv^{-1}$) denotes the word $u$ such that $vu =w$ (resp. $ uv=w$). If $v$ is a prefix of $w$ we shall write $v\preceq w$ and, if $v\neq w$, $v\prec w$.

A word $w$ is called \emph{primitive} if $w\neq v^n$, for all $v\in A^*$ and $n>1$. We let
$\PR$ denote the set of all primitive words of $A^*$.

 The \emph{reversal} of a word $w=a_1a_2\cdots a_n$, with $a_i\in A$, $1\leq
i\leq n$, is the word $ w^{\sim}=a_n\cdots a_1$. One sets
$\varepsilon^{\sim}=\varepsilon$. A \emph{palindrome} is a word which equals
its reversal. The set of all palindromes over $A$ will be denoted by $\PAL(A)$, or $\PAL$ when no confusion
arises. For any $X\subseteq A^*$ we 
 set $X^{\sim} = \{ x^{\sim} \mid x\in X \}$.
For any word $w\in A^*$  we let $LPS(w)$ denote the longest palindromic suffix of $w$. For  $X\subseteq A^*$, we set $LPS(X)= \{LPS(x) \mid x\in X\}$. A word $w$ is said to be \emph{rich} in palindromes, or simply rich, if it has the maximal possible number of distinct palindromic factors, namely $|w|+1$ (cf. \cite{DJP}).

A  right infinite word, or simply \emph{infinite word},  $w$ is just an infinite sequence of letters:
\[w=a_1a_2\cdots a_n\cdots \ \ ,\text{where }a_i\in A,\,\text{ for all } i\geq 1\enspace.\]
For any integer $n\geq 0$, $w_{[n]}$ will denote the prefix $a_1a_2\cdots a_n$ of $w$ of length $n$.
A factor of $w$ is either the empty word or any sequence  $a_i\cdots a_j$ with $i\leq j$. 
If $w=uvvv\cdots v\cdots =uv^{\omega}$ with $u\in A^*$ and $v\in A^+$, then $w$ is called \emph{ultimately periodic} and \emph{periodic}
if $u=\varepsilon$.

The set of all infinite words over $A$ is denoted by $A^{\omega}$. We also set $A^{\infty}= A^*\cup A^{\omega}$.
For any $w\in A^{\infty}$  we denote respectively by $\Ff w$ and
 $\Pre w$
the sets of all factors and
 prefixes
 of the word $w$. For $X\subseteq A^*$, $\Pre X$ denotes the set of all prefixes of the words of $X$.
 
 Let $w\in A^\infty$. A factor $u$ of $w$ is  \emph{right special} (resp. \emph{left special}) if there exist two letters $a,b\in A$, $a\neq b$, such that $ua$ and $ub$ (resp.  $au$ and $bu$) are factors of $w$. The factor $u$ is called \emph{bispecial} if it is right and left special. The \emph{order} of a right (resp. left) special factor $u$ of $w$ is the number of distinct letters $a\in A$ such that $ua \in \Ff w$ (resp. $au \in \Ff w)$.
 
Let $w\in  A^\infty$ and $u$ a factor of $w$.  An {\em occurrence} of  $u$   in $w$ is any $\lambda \in A^*$ such that  $\lambda u \preceq w$. If $\lambda_1$ and $\lambda_2$ are two distinct occurrences of $u$  in $w$ with $|\lambda_1| < |\lambda_2|$, the gap between the
 occurrences is $|\lambda_2|- |\lambda_1|$. For any $w\in A^*$ and letter $a\in A$,  $|w|_a$ denotes the number of occurrences of the letter $a$ in $w$. 
 
 The \emph{factor complexity} $p_w$ of a word $w\in A^{\infty}$ is the map $p_w:\Nn \rightarrow \Nn$
 counting for each $n\geq 0$ the distinct factors of $w$ of length $n$, i.e.,
 \[ p_w(n) = \card(A^n \cap \Ff w).\]
 The following recursive formula (see, for instance, \cite{adl99}) allows one to compute the factor complexity in terms of right special factors: for all $n\geq 0$
 \begin{equation}\label{eq:rsf}
 p_w(n+1) = p_w(n) + \sum_{j=0}^d(j-1)s_r(j,n),
 \end{equation}
 where $d=\card(A)$, and $s_r(j,n)$ is the number of right special factors of $w$ of length $n$ and order $j$.
 
A \emph{morphism} (resp.~\emph{antimorphism}) from $A^*$ to the free monoid
$B^*$ is any map $\varphi:A^*\to B^*$ such that $\varphi(uv)=\varphi(u)\varphi(v)$
(resp.~$\varphi(uv)=\varphi(v)\varphi(u)$) for all $u,v\in A^*$. A morphism $\varphi$ can be naturally extended to
$A^\omega$ by setting for any $w=a_1a_2\cdots a_n\cdots\in A^\omega$,
\[\varphi(w)=\varphi(a_1)\varphi(a_2)\cdots\varphi(a_n)\cdots\;.\]

A \emph{code} over $A$ is a subset $X$ of $A^+$ such that every word of $X^+$
admits a unique factorization by the elements of $X$ (cf.~\cite{codes}). A
subset of $A^+$ with the property that none of its elements is a proper prefix
(resp.~suffix) of any other is trivially a code, usually called 
\emph{prefix} (resp.~\emph{suffix}). 
We recall that if $X$ is a prefix (resp. suffix)
code, then $X^*$ is \emph{right unitary} (resp. \emph{left unitary}), i.e., for all $p\in X^*$ and $w\in
A^*$, $pw\in X^*$ (resp. $wp \in X^*$) implies $w\in X^*$. 

A \emph{bifix} code is a code which
is both prefix and suffix. A code  $X$ is  called \emph{infix} if no word of $X$ is a proper
factor of another word of $X$. A code $X$ will be  called \emph{weakly overlap-free} if no word $x\in X$ can be
factorized as $x=sp$ where $s$ and $p$ are respectively a proper non-empty suffix of a word $x'\in X$ and a proper non-empty prefix of a word $x''\in X$. Note that the code $X=\{abb, bbc\}$ is not overlap-free \cite{BdD}, but it is weakly overlap free.

A code $X$ has a \emph{finite deciphering delay} if there exists an integer $k$ such that for all $x, x' \in X$, if  $x X^kA^* \cap x' X^*\neq \emptyset$ then $x =x' $. The minimal $k$ for which the preceding condition is satisfied is called deciphering delay of $X$. A prefix code has a deciphering delay equal to 0.

Let  $X$ be a set of words over $A$. We let $X^{\omega}$ denote the set of all infinite
words
\[ x = x_1x_2 \cdots x_n \cdots , \mbox{with} \ x_i\in X, \ i\geq 1.\]
As is well known \cite{codes}, if $X$ is a code having a finite deciphering delay, then any  $x\in X^{\omega}$ can be uniquely factorized by the elements of $X$.

\subsection{The palindromization map}

 We  introduce in
$A^*$ the map $^{(+)} : A^*\rightarrow \PAL$ which
associates to any word $w\in A^*$ the palindrome $w^{(+)}$
defined as the shortest palindrome having the prefix $w$ (cf. \cite{deluca}).  We call
$w^{(+)}$ the \emph{right palindromic closure of} $w$.  If $Q=LPS(w)$ is the
longest palindromic suffix of $w= uQ$, then one has
\[
	w^{(+)}=uQu^{\sim}\,.
\]
Let us now define the map
\[
	\psi: A ^*\rightarrow \PAL ,
\]
called \emph{right iterated palindromic closure}, or simply \emph{palindromization map}, over $A^*$,  as follows: $\psi(\varepsilon)=\varepsilon $ and for all
$u\in A ^*$, $a\in A $,
\[
	\psi(ua)=(\psi(u)a)^{(+)}\,.
\]

The following proposition summarizes some simple but noteworthy  properties of  the palindromization map
 (cf., for instance, \cite{DJP, deluca}):
\begin{prop}\label{prop:basicp} The palindromization map $\psi$ over $A^*$  satisfies the following
properties: for $u,v\in A^*$
\begin{itemize}
\item[P1.]  If $u$ is  a prefix of  $v$, then $\psi(u)$ is a palindromic prefix (and suffix) of $\psi(v)$.
\item[P2.] If $p$ is a prefix of $\psi(v)$, then $p^{(+)}$ is a prefix of $\psi(v)$.
\item[P3.] Every palindromic prefix of $\psi(v)$ is of the form $\psi(u)$ for some prefix $u$ of $v$.
\item[P4.] The   palindromization map  is  injective.
\end{itemize}
\end{prop}
 For any  $w\in \psi(A^*)$ the unique word  $u$ such that  $\psi(u)=w$ is called the \emph{directive word} of $w$.
One can extend
  $\psi$  to $A^{\omega}$ as follows: let  $w\in A^{\omega}$ be an infinite word
\[ w = a_1a_2\cdots a_n\cdots, \ \  \ a_i\in A, \ i\geq 1.\]
Since by property P1 of the preceding proposition  for all $n$, $\psi(w_{[n]})$ is a prefix of  $\psi(w_{[n+1]})$,  we  can define  the infinite word $\psi(w)$ as:
\[ \psi(w) = \lim_{n\rightarrow \infty} \psi(w_{[n]}).\]
The extended map $\psi: A^{\omega}\rightarrow A^{\omega}$ is injective. The word $w$ is called the \emph{directive word} of $\psi(w)$. 

The family of infinite words $\psi(A^{\omega})$ is the class of the \emph{standard episturmian words}, or simply \emph{epistandard words}, over $A$
introduced in \cite{DJP}(see also \cite{JP}).  When each letter of $A$  occurs infinitely often in the directive word,  one has the class of the \emph{standard Arnoux-Rauzy words} \cite{AR, Rauzy}. A standard Arnoux-Rauzy word over a binary alphabet is usually called \emph{standard Sturmian word}. $Epistand_A$ will denote the class of all epistandard words over $A$.

An infinite word $s\in A^{\omega}$ is called \emph{episturmian} (resp. {\em  Sturmian}) if there exists a standard episturmian (resp. Sturmian) word $t\in A^{\omega}$ such that $\Ff s = \Ff t$.

 The words of the set $\psi(A^*)$ are the palindromic prefixes of all standard episturmian words over the alphabet $A$. They  are called \emph{epicentral words}, and simply  \emph{central} \cite{LO2},  in the case of a two-letter alphabet.

 \begin{example}  Let $A=\{a,b\}$. If $w = (ab)^{\omega}$, then the standard Sturmian word $f=\psi( (ab)^{\omega})$ having the directive word $w$ is the famous \emph{Fibonacci word}
 \[ f = abaababaabaab\cdots\]
 In the case of a three letter alphabet $A=\{a,b,c\}$ the standard Arnoux-Rauzy word having the directive
 word $w =(abc)^{\omega}$ is the so-called \emph{Tribonacci word}
 \[ \tau = abacabaabacaba\cdots.\]
  \end{example}

  \section{A generalized palindromization map}\label{sec:three}
  
  Let $X$ be a code over the alphabet $A$. Any word $w\in X^+$ can be uniquely factorized
  in terms of the elements of $X$. So we can introduce the map \[\psi_X : X^* \rightarrow PAL,\]
  inductively defined for any $w\in X^*$ and $x\in X$ as:
  \[\psi_X(\varepsilon)= \varepsilon,  \ \psi_X(x) = x^{(+)},\]
  \[ \psi_X(wx) = (\psi_X(w)x)^{(+)}.\]
  In this way to each word $w\in X^*$, one can uniquely associate the palindrome $\psi_X(w)$. We call $\psi_X$ the \emph{palindromization map relative to the code} $X$. 
   If $X=A$, then $\psi_A = \psi$.

  \begin{example} Let $A=\{a, b\}$, $X= \{ab, ba\}$, and $w=abbaab$; $X$ is a code so that
  $w$ can be uniquely factorized as  $w= x_1x_2x_1$ with $x_1=ab$ and $x_2=ba$. One has:
  $\psi_X(ab) = aba$, $\psi_X(abba)= (ababa)^{(+)}= ababa$, and $\psi_X(abbaab)=ababaababa$.
  \end{example} 
  
 The properties of the palindromization map $\psi$ stated in Proposition \ref{prop:basicp} are not in general satisfied by the generalized palindromization map $\psi_X$.
 For instance, take $X=\{ab, abb\}$ one has  $ab \prec abb$ but  $\psi_X(ab)= aba$ is not a prefix of
 $\psi_X(abb)= abba$. Property P1 can be replaced by the following:
 
 \begin{prop}\label{prop:pref} Let $v= x_1\cdots x_n$ with $x_i\in X$, $i=1,\ldots, n$. For any $v_j=x_1\cdots x_j$, $1\leq j < n$ one has $\psi_X(v_j) \prec \psi_X(v)$. If  $X$ is a prefix code, then the following holds:  for $u, v\in X^*$ if $u\preceq v$,
 then $\psi_X(u) \preceq \psi_X(v)$. 
 \end{prop}
 \begin{proof} For any $j=1,\ldots, n-1$ one has 
 \[ \psi_X(x_1\cdots x_jx_{j+1}) = (\psi_X(x_1\cdots x_j)x_{j+1})^{(+)},\]
 so that $\psi_X(x_1\cdots x_j) \prec \psi_X(x_1\cdots x_{j+1})$. From the transitivity of relation $\prec$ it follows
  $\psi_X(v_j) \prec \psi_X(v)$. Let now $X$ be a prefix code and suppose that $u, v\in X^*$ and
  $u\preceq  v$. We can write $v= x_1\cdots x_n$ and $u= x'_1\cdots x'_m$ with $x_i, x'_j\in X$, $i=1,\ldots, n$ and $j=1, \ldots, m$. Since   $u\preceq  v$, one has $v=u\zeta$, with $\zeta \in A^*$. From the right unitarity of $X^*$ it follows $\zeta\in X^*$ and, therefore, $x'_i=x_i$ for $i=1, \ldots, m$. From the preceding result it follows that $\psi_X(u) \preceq \psi_X(v)$. 
   \end{proof}
   
   Properties P2  and P3 are also  in general not satisfied by $\psi_X$. As regards P2,  consider, for instance, the code
   $X=\{a, ab, bb\}$ and the word $w= abbab$. One has $\psi_X(w)= abbaabba$. Now $ \psi_X(w)$ has the prefix $ab$  but not  $(ab)^{(+)} =aba$. As regards P3 take $X= \{abab,b\}$ one has that $\psi_X(abab)= ababa$. Its palindromic prefix $aba$ is not equal to $\psi_X(v)$ for any $v\in X^*$.

  Differently from $\psi$, the map $\psi_X$ is not in general injective. For instance, if $X$ is the code $X= \{ab, aba\}$, then $\psi_X(ab)=\psi_X(aba)= aba$. Property P4  can be replaced by the following:
  
  \begin{prop}\label{prop:prefcode}Let $X$ be a prefix code over $A$. Then $\psi_X$ is injective.
  \end{prop}
   \begin{proof} Suppose that there exist words $x_1,\ldots, x_m, x'_1, \ldots x'_n \in X$ such that
  \[  \psi_X(x_1\cdots x_m)=  \psi_X(x'_1\cdots x'_n).\]
  We shall prove that $m=n$ and that for all $1\leq i \leq n$, one has $x_i=x'_i$.
  
     Without loss of generality, we can suppose $m\leq n$. Let us first prove by induction that for all $1\leq i \leq m$, one has $x_i=x'_i$.  Let us assume that $x_{1}=x'_{1}$, \ldots, $x_{k}=x'_{k}$ for $0<k<m$ and
   show that  $x_{k+1}= x'_{k+1}$. To this end 
   let us set $w=\psi_X(x_1\cdots x_m)$ and $w'=\psi_X(x'_1\cdots x'_m)$. In view of the preceding proposition, we can write:
  \[ w = \psi_X(x_1\cdots x_kx_{k+1})\zeta = (\psi_X(x_1\cdots x_k)x_{k+1})^{(+)}\zeta\]
  and
   \[ w' = \psi_X(x_1\cdots x_kx'_{k+1})\zeta' = (\psi_X(x_1\cdots x_k)x'_{k+1})^{(+)}\zeta',\]
   with $\zeta, \zeta'\in A^*$. Now one has:
  \[ (\psi_X(x_1\cdots x_k)x_{k+1})^{(+)}= \psi_X(x_1\cdots x_k)x_{k+1}\xi\]
  and
  \[ (\psi_X(x_1\cdots x_k)x'_{k+1})^{(+)}= \psi_X(x_1\cdots x_k)x'_{k+1}\xi',\]
  with $\xi, \xi'\in A^*$. Therefore, we obtain:
  \[ w= \psi_X(x_1\cdots x_k)x_{k+1}\xi\zeta= \psi_X(x_1\cdots x_k)x'_{k+1}\xi' \zeta'=w'.\]
  By cancelling on the left in both the sides of previous equation the common prefix $\psi_X(x_1\cdots x_k)$ one derives
  \begin{equation}\label{eq:pref00}
  x_{k+1}\xi\zeta= x'_{k+1}\xi' \zeta'.
  \end{equation}
  Since $X$ is a prefix code one obtains $x_{k+1}=x'_{k+1}$. Since an equation similar to (\ref{eq:pref00})
  holds also in the case $k=0$ one has also $x_1=x'_1$. Therefore, $x_i=x'_i$ for $i=1,\ldots,m$.
   We can write:
  \[ \psi_X(x_1\cdots x_m)= \psi_X(x_1\cdots x_mx'_{m+1}\cdots x'_n).\]
  Since by Proposition \ref{prop:pref},  $ \psi_X(x_1\cdots x_m) \preceq \psi_X(x_1\cdots x_mx'_{m+1}\cdots x'_n)$ it follows that $m=n$.
   \end{proof}

   A partial converse of the preceding proposition is:
   
   \begin{prop} Let $X$ be a code such that  $X\subseteq \PAL \cap \PR$. If $\psi_X$ is injective, then $X$ is  prefix.
   \end{prop}
   \begin{proof} Let us suppose that $X$ is not a prefix code. Then there exist words $x,y \in X$
   such that $x\neq y$ and  $y= x\lambda$ with $\lambda\in A^+$. Since $x,y\in \PAL$ one has
   $y = x \lambda  =  \lambda^{\sim} x$.
   We shall prove that the longest palindromic suffix $LPS(yyx)$ of the word $yyx = \lambda^{\sim}xyx$ is
   $xyx$. This would imply, as $x,y\in \PAL$,  that
   \[ \psi_X(yyx)= (yyx)^{(+)} =  \lambda^{\sim}xyx\lambda = yyy = (yyy)^{(+)}= \psi_X(yyy),\]
   so that $\psi_X$ would be not injective, a contradiction.
   
   Let us then suppose that $y = \lambda^{\sim }x = \alpha x \beta$, $\alpha, \beta \in A^*$, and that
   $ LPS(yyx) = x\beta yx.$
   This implies  $\beta y\in \PAL$, so that,
   $ \beta y = \beta \alpha x \beta = y\beta^{\sim}= \alpha x \beta \beta^{\sim}$.
   Therefore, one has $\beta= \beta^{\sim}$ and
   \[  \beta (\alpha x \beta) = (\alpha x \beta) \beta.\]
   From a classic result of combinatorics on words \cite{LO},  there exist $w \in \PR$ and integers $h,k\in \Nn$ such that $\beta = w^h$ and $y = \alpha x \beta = w^k$. Since $y\in \PR$, it follows that $k=1$, $y=w$, and
   $\beta = y^h$. As $|\beta|< |y|$, the only possibility is $h=0$, so that $\beta= \varepsilon$, which
   implies $LPS(yyx) = xyx$.
   \end{proof}
   
   \section{An extension to infinite words}\label{sec:four}
  
  Let us now consider a code $X$ having a finite deciphering delay. One can extend $\psi_X$ to
  $X^{\omega}$ as follows: let   $ x = x_1x_2 \cdots x_n \cdots ,$ with $x_i\in X,  i\geq 1$. From Proposition \ref{prop:pref}, for any
  $n\geq 1$, $\psi_X(x_1\cdots x_n)$ is a proper prefix of $\psi_X(x_1\cdots x_nx_{n+1})$ so that
  there exists
  			\[ \lim_{n\rightarrow \infty} \psi_X(x_1\cdots x_n) = \psi_X(x).\]
Let us observe that the word 	$\psi_X(x)	$ has infinitely many palindromic prefixes. This implies that 	
$\psi_X(x)	$ is \emph{closed under reversal}, i.e., if $w\in \Ff  \psi_X(x)$, then also $w^{\sim}\in \Ff  \psi_X(x)$.
If $X=A$ one obtains the usual extension of $\psi$ to the infinite words.

Let us explicitly remark that if $X$ is a code with an infinite deciphering delay one cannot associate by the generalized palindromization map to each word $x\in X^{\omega}$ a unique infinite word. For instance, the code  $X=\{a, ab, bb\}$ has an infinite deciphering delay; the word
$ab^{\omega}$ admits two distinct factorizations by the elements of $X$. The first beginning with
$ab$ is  $(ab)(bb)^{\omega}$, the second beginning with $a$ is $a(bb)^{\omega}$. Using the first decomposition one can generate by the generalized palindromization map the infinite word
$(ababb)^{\omega}$ and using the second the infinite word $(abb)^{\omega}$. 

Let us observe that the previously defined map $\psi_X: X^{\omega} \rightarrow A^{\omega}$ is not in general injective. For instance, take the code $X=\{ab, aba\}$ which has finite deciphering delay equal to $1$. As it is readily verified one has $\psi_X((ab)^{\omega}) = \psi_X((aba)^{\omega})=(aba)^{\omega}$.

The following proposition holds; we omit its proof, which is very similar to that of Proposition \ref{prop:prefcode}.

\begin{prop}\label{prop:prefcode1} Let $X$ be a prefix code over $A$. Then the map $\psi_X: X^{\omega}\rightarrow A^{\omega}$ is injective.
\end{prop}

The class of infinite words that one can generate by means of  generalized palindromization maps $\psi_X$ is, in general, strictly larger than the class of standard episturmian words.

\begin{example} Let $A= \{a,b\}$ and $X= \{a, bb\}$. Let $x$ be any infinite word  $x= abbay$ with
$y\in X^{\omega}$. One has that $\psi_X(abba)= abbaabba$, so that the word $\psi_X(x)$ will not
be balanced (cf. \cite{LO2}). This implies that $\psi_X(x)$ is not  a  Sturmian word.
Let $A=\{a, b, c\}$ and $X= \{a, abca\}$. Take any word $x= abcay$ with $y\in X^{\omega}$. One
has $\psi_X(abca)= abcacba$. Since the prefix $abca$ is not rich in palindromes, it follows that
$\psi_X(x)$ is not an episturmian word.
\end{example}	

\begin{thm}\label{thm:unifrecur}
For any finite code $X$ having finite deciphering delay and any 
$t\in X^{\omega}$, the word $s=\psi_{X}(t)$ is uniformly recurrent. 
\end{thm}
\begin{proof}
Let $t= x_1x_2\cdots x_n \cdots \in X^{\omega}$, with $x_i\in X$, $i\geq 1$, and $w$ be any factor of $s$. Let $\alpha$ be the shortest prefix $\alpha = x_1\cdots x_h$  of
$t$ such that $w\in\Ff u$, with $u=\psi_{X}(\alpha)$. The word $s$ is trivially recurrent since it has infinitely many palindromic prefixes. Hence, $w$ occurs infinitely many times in $s$. We will show that the gaps between successive occurrences of $w$ in $s$ are bounded above by $|u|+2\ell_X$, where $\ell_X=\max_{x\in X}|x|$.
This is certainly true within the prefix $u$: even if $w$ occurs in $u$ more than once, the gap between any two such occurrences cannot be longer than $|u|$.

Let us then assume we proved such bound on gaps for successive occurrences
of $w$ in $\psi_{X}(\beta)$, where $\beta = x_1\cdots x_k$, $h\leq k$, and let us prove it for occurrences in $\psi_{X}(\beta y)$, where $y= x_{k+1}$. 
We can write $\psi_{X}(\beta)=u\rho=\rho^{\sim}u$ and $\psi_{X}(\beta y)=\psi_{X}(\beta)\lambda=\lambda^{\sim}\psi_{X}(\beta)$ for some $\lambda,\rho\in A^{*}$, so that 
\begin{equation}
\label{eq:betay}
\psi_{X}(\beta y)=\rho^{\sim}u\lambda=\lambda^{\sim}u\rho\;.
\end{equation}
By inductive hypothesis, the only gap we still need to consider is the one between the last occurrence of $w$ in $\rho^{\sim}u$ and the first one in $u\rho$ as displayed in~\eqref{eq:betay}. 
If $|\rho|>|\lambda|$, then both such occurrences of $w$ fall within  $\rho^{\sim}u= \psi_X(\beta)$, so that by induction we are done. So suppose $|\lambda|>|\rho|$.
As one easily verifies, the previous gap is at most equal to the gap between the two displayed occurrences of $u$ in~\eqref{eq:betay},  namely  $|\lambda|-|\rho|$. From ~\eqref{eq:betay} one has:
\[|\lambda|-|\rho|=|\psi_X(\beta y)|-|\psi_X(\beta)|-(|\psi_X(\beta)|-|u|)=|\psi_{X}(\beta y)|-2|\psi_{X}(\beta)|+|u|.\]
Now, as
\[|\psi_{X}(\beta y)|=|(\psi_{X}(\beta)y)^{(+)}|<2(|\psi_{X}(\beta)|+|y|)\leq 2|\psi_{X}(\beta)|+2\ell_X\;,\]
we have $|\lambda|-|\rho|<|u|+2\ell_X$. By induction, we can conclude that gaps between successive occurrences of $w$ are bounded by $|u|+2\ell_X$ in the whole $s$, as desired.
\end{proof}

Let $y= y_1y_{2}\cdots y_n\cdots \in X^{\omega}$, with $y_i\in X$ for all $i\geq 1$. We say that a word $x\in X$ is \emph{persistent} in $y$ if there exist infinitely many integers $i_1<i_2< \cdots < i_k < \cdots$ such that $x= y_{i_k}$ for all $k\geq 1$.

We say that the word $y= y_1y_{2}\cdots y_n\cdots \in X^{\omega}$ is \emph{alternating} if there exist distinct letters $a,b\in A$, a word 
$\lambda\in A^{*}$, and a sequence of indices $i_0<i_1< \cdots < i_n < \cdots$, such that $\lambda a\preceq y_{i_{2k}}$ and $\lambda b\preceq  y_{i_{2k+1}}$ for all $k\geq 0$.

We remark that if there exist two distinct words $x_{1},x_{2}\in X$, which are persistent in $y$ and such that $\{x_{1},x_{2}\}$ is a prefix code, then $y$ is alternating. If $X$ is finite, then the two conditions are actually equivalent.

\begin{prop}\label{prop:alter}
Let  $y=y_{1}\cdots y_{n}\cdots\in X^{\omega}$ with $y_i\in X$, $i\geq 1$. If $y$ is alternating, then $\psi_{X}(y)$ is not ultimately periodic.
\end{prop}
\begin{proof}
By hypothesis, there exists an increasing sequence of indices
$(i_{n})_{n\geq 0}$, such that for all $k\geq 0$ we have
$\lambda a \preceq y_{i_{2k}}$ and $\lambda b\preceq y_{i_{2k+1}}$, for some
$\lambda\in A^{*}$ and letters $a\neq b$.

For all $n\geq 0$, let $u_{n}$ denote the word $\psi_{X}(y_{1}\cdots y_{n})$. We shall prove that $u_{n}\lambda$ is a right special factor of
$s=\psi_{X}(y)$ for any $n$, thus showing that $s$ cannot be ultimately periodic (cf. \cite{LO2}).

We can choose an integer $h>0$ satisfying $i_{2h}>n$. Let us set $m = i_{2h}$ and 
$x_1= y_{i_{2h}}$. Now one has that:
\[ u_{m-1}x_1 \preceq u_m \in \Pre s .\]
Since $u_n$ is a prefix and a suffix of  $u_{m-1}$ it follows,  writing $x_1 = \lambda a \eta$ for some $\eta\in A^*$, that
\[ u_nx_1= u_n \lambda a \eta \in \Ff s.\]
Since $i_{2h+1}>i_{2h}$, setting $x_2= y_{2h+1}= \lambda b \eta'$ for some $\eta'\in A^*$, one derives by a similar argument that:
\[ u_nx_2= u_n \lambda b \eta' \in \Ff s.\]
From the preceding equations one has that $u_n\lambda$ is a right special factor of $s$.
\end{proof}

We shall now prove a theorem showing how one can generate all standard Sturmian words by
the palindromization map relative to the code $X= \{a,b\}^2$. We premise the		 following lemma  which is  essentially a restatement of a well known characterization of central words (see for instance ~\cite[Proposition~9]{deluca}).
\begin{lemma}
\label{thm:PxyQ}
Let $A=\{a,b\}$ and $E$ be the automorphism of $A^*$ interchanging  the letter $a$ with $b$. If $z\in A$ and $w\in A^{*}\setminus z^{*}$, then
\[\psi(wz)=\psi(w)z E(z) \psi(w') \  \mbox{ for some} \  w'\in\Pre w. \]
\end{lemma}

\begin{thm}
	Let $A=\{a,b\}$ and $X=A^{2}$. An infinite word $s\in A^{\omega}$ is standard
	Sturmian if and only if $s=\psi_{X}(t)$ for some alternating $t\in X^{\omega}$ 
	such that
	\[t\in \left((aa)^{*}\cup(bb)^{*}\right)\{ab,ba\}^{\omega}\;.\]
\end{thm}
\begin{proof}
Let $s=\psi_{X}(t)$; we can assume without loss of generality that $t\in (aa)^{k}\{ab,ba\}^{\omega}$ with $k\in \Nn$. Let $t_{[2n]}$ be the prefix of $t$ of length $2n$ (which belongs to $X^{*}$). We shall prove that  $\psi_X(t_{[2n]})$ is a central word for all $n\geq 0$. This is trivial for all prefixes $t_{[2p]}$ of $t$ with
$p\leq k$. Let us now assume, by induction, that $\psi_X(t_{[2n]})$ is central for a given $n\geq k$
and prove that  $\psi_X(t_{[2n+2]})$ is central.

We can write $t_{[2n+2]}=t_{[2n]}ab$ or $t_{[2n+2]}=t_{[2n]}ba$. Since by the inductive hypothesis
$\psi_X(t_{[2n]})$ is central, there exists $u_n\in A^*$ such that $\psi_X(t_{[2n]})= \psi(u_n)$. The words
$\psi_X(t_{[2n]})ab$ and $\psi_X(t_{[2n]})ba$ are finite standard words and therefore, as is well known, prefixes of standard Sturmian words (cf. \cite[Corollary~2.2.28]{LO2}). By property P2  of Proposition \ref{prop:basicp}, their palindromic closures 
$(\psi_X(t_{[2n]})ab)^{(+)}= \psi_X(t_{[2n]}ab)$ and $(\psi_X(t_{[2n]})ba)^{(+)}=\psi_X(t_{[2n]}ba)$ are both central. Hence, in any case $\psi_X(t_{[2n+2]})$ is central so that there exists $u_{n+1}\in A^*$ such that
$\psi_X(t_{[2n+2]})= \psi(u_{n+1})$. Since $\psi(u_n)$ is a prefix of $\psi(u_{n+1})$ from Proposition \ref{prop:basicp}  one derives that
$u_{n} \prec u_{n+1}$.

We have thus proved the existence of a sequence of finite words $(u_{n})_{n\geq 0}$, with $u_{i}\prec u_{i+1}$ for all $i\geq 0$, such that for all $n\geq 0$ we have 
\[\psi_{X}(t_{[2n]})=\psi(u_{n})\;.\]
 Letting $\Delta=\lim_{n\to\infty} u_{n}$, we obtain $s=\psi(\Delta)$. Since $t$ is alternating, $s$ is not ultimately periodic by Proposition~\ref{prop:alter}, so that it is a standard Sturmian word.

Conversely, let $s$ be a standard Sturmian word, and let $\Delta$ be its directive word. Without loss of generality, we can assume that $\Delta$ begins in $a$; let $n\geq 1$ be such that $a^{n}b\in\Pre \Delta$. If $n$ is even, we have \[\psi(a^{n}b)=\left((aa)^{\frac{n}{2}}b\right)^{(+)}=
\left((aa)^{\frac{n}{2}}ba\right)^{(+)}=\psi_{X}\left((aa)^{\frac{n}{2}}ba\right)\]
whereas if $n$ is odd,
\[\psi(a^{n}b)=\left((aa)^{\frac{n-1}{2}}ab\right)^{(+)}=
\psi_{X}\left((aa)^{\frac{n-1}{2}}ab\right)\;.\]
Let now $z\in A$ and $uz$ be a prefix of $\Delta$ longer than $a^{n}b$. By induction, we can suppose that there exists some $w\in (aa)^{*}\{ab,ba\}^{*}$ such that $\psi(u)=\psi_{X}(w)$. From Lemma~\ref{thm:PxyQ} and Proposition \ref{prop:basicp}, we obtain, setting $\hat z=E(z)$, $(\psi(u)z\hat z)^{(+)} \preceq \psi(uz) \preceq (\psi(u)z\hat z)^{(+)}$. Hence,
\begin{equation}
\label{eq:deltat}
\psi(uz)=(\psi(u)z\hat z)^{(+)}=(\psi_{X}(w)z\hat z)^{(+)}=\psi_{X}(wz\hat z)\;.
\end{equation}
We have thus shown how to construct arbitrarily long prefixes of the desired infinite word $t$, starting from the Sturmian word $s$. Since $a$ and $b$ both occur infinitely often in $\Delta$, by~\eqref{eq:deltat} we derive that $t$ is alternating.
\end{proof}

\begin{example} In the case of Fibonacci word $f$ let us take $X=\{ab, ba\}$.
 As it is readily verified, one has:
  \[f= \psi_X(ab(abba)^{\omega}).\]
\end{example}

\vspace{4 mm}

Let $\mu$ be the Thue-Morse morphism, and $t=\mu^{\omega}(a)$ the Thue-Morse 
word \cite{LO}. We recall that $\mu$ is defined by $\mu(a)=ab$ and $\mu(b)=ba$. The next proposition will show that $t$ can be obtained using our 
generalized palindromization map, relative to a suitable infinite code.

Let us set $u_{n}=\mu^{2n}(a)$ and $v_{n}=E(u_{n})b$, for all $n\in\Nn$. 
Thus $v_{0}=bb$, $v_{1}=baabb$, $v_{2}=baababbaabbabaabb$, and so on.
\begin{prop}
	The set $X=\{a\} \cup \{v_{n} \mid n\in\Nn\}$ is a prefix code, and
	\[t=\psi_{X}(av_{0}v_{1}v_{2}\cdots)\;.\]
\end{prop}
\begin{proof}
	As a consequence of~\cite[Theorem~8.1]{adlADL}, we can write
	$u_{n+1}=\mu^{2n+2}(a)=\left(\mu^{2n+1}(a)b\right)^{(+)}$.
	Since for any $k\geq 0$ one has $\mu^{k+1}(a)=\mu^{k}(a)E\left(\mu^{k}(a)
	\right)$, we obtain for all $n\geq 0$
	\begin{equation}
		\label{eq:un1}
		u_{n+1}=(u_{n}E(u_{n})b)^{(+)}=(u_{n}v_{n})^{(+)}\;.
	\end{equation}
	Since $b \prec v_{i}$ for all $i\geq 0$, by~\eqref{eq:un1} it follows
	$u_{i}b\prec u_{i}v_{i} \preceq u_{i+1}$, so that $u_{i}b	\prec u_{j}$ 
	whenever $0\leq i<j$, whence
         $E(u_{i}b)=E(u_{i})a \prec E(u_{j})$.
	This implies that for $0\leq i<j$, $v_{i}=E(u_{i})b$ is not a prefix 
	of $v_{j}=E(u_{j})b$. Clearly $v_{i}$ is not a prefix of any $v_{k}$ with 
	$k<i$, nor of $a$, which in turn is not a prefix of any $v_{i}$ with $i\in
	\Nn$; hence $X$ is a prefix code.
	
	Since $u_{0}=a=\psi_{X}(a)$, from~\eqref{eq:un1} it follows that for all
	$n>0$, $u_{n}=\psi_{X}(av_{0}\cdots v_{n-1})$. As $t=\lim_{n\to\infty}
	u_{n}$, the assertion is proved.
\end{proof}

\section{Generalized Arnoux-Rauzy words}\label{sec:five}

Let us  suppose that the code $X$ over the alphabet $A$ is finite and \emph{maximal}, i.e., it is not properly included in any other code on the same alphabet. By a classic result of Sch\"utzenberger either $X$ is prefix or has an infinite deciphering delay \cite{codes}. Therefore, if one wants to define a map $\psi_X: X^{\omega}\rightarrow A^{\omega}$ one has to suppose that the code is a prefix maximal code.

We shall now introduce a class of infinite words which are a natural generalization in our framework of the standard Arnoux-Rauzy words.

Let $X$ be a finite maximal prefix code over the alphabet $A$ of cardinality $d>1$. We say that the word $s= \psi_X(y)$,
with $y\in X^{\omega}$ is a standard  \emph{Arnoux-Rauzy word relative to $X$}, or $X$-AR word for short, if every  word $x\in X$ is persistent in $y$.

Let us observe that if $X=A$ we have the usual definition of standard Arnoux-Rauzy word. Any $X$-AR word is trivially alternating and therefore, from Proposition \ref{prop:alter} it is not ultimately periodic.  The following proposition extends to $X$-AR words a property satisfied by the classic standard Arnoux-Rauzy words.

\begin{prop} Let $s= \psi_X(y)$ be an $X$-AR word  with $y=y_1\cdots y_n\cdots$, $y_i\in X$, $i\geq 1$. Then for any $n\geq 0$,
$u_n=\psi_X(y_1\cdots y_n)$ is a bispecial factor of $s$ of order $d=\card (A)$. This implies that  every prefix of $s$ is a left special factor of $s$ of order $d$.
\end{prop}
\begin{proof} Since $X$ is a finite maximal prefix code, it is complete \cite{codes}, i.e., it is represented by the leaves of a full $d$-ary tree (i.e., each node in the tree is either a leaf or has exactly degree $d$).  Hence, $X^f = A$, where $X^f$ denotes the set formed by the first letter of all words of $X$.  Any word $x\in X$ is persistent in $y$, so that, by using  an argument similar to that of the proof of Proposition \ref{prop:alter},  one has that for any $n\geq 0$,
\[ u_nX \subseteq \Ff s,\]
that implies $u_nX^f = u_nA \subseteq \Ff s$, i.e., $u_n$ is a right special factor of $s$ of order $d$.
Since $s$ is closed under reversal and $u_n$ is a palindrome, one has that $u_n$ is also a left special factor of $s$ of order $d$. Hence, $u_n$ is a bispecial factor of order $d$.
Let $u$ be a prefix of $s$. There exists an integer $n$ such that $u\preceq u_n$. From this one has that $u$ is a left special factor of $s$ of order $d$.
\end{proof}

An infinite word $s$ over the alphabet $A$ is \emph{$\omega$-power free} if for every non-empty word
$u\in \Ff s$ there exists an integer  $p>0$ such that $u^p \not\in \Ff s$. We recall the following result (see, for instance, ~\cite[Lemma~2.6.2]{dV})
which will be useful in the sequel:
\begin{lemma}\label{lem:uromega} A uniformly recurrent word is either periodic or $\omega$-power free.
\end{lemma}
\begin{cor} An $X$-AR word is $\omega$-power free.
\end{cor}
\begin{proof} An $X$-AR word is not periodic and by Theorem \ref{thm:unifrecur}  it is uniformly recurrent, so that the result  follows from the preceding  lemma.
\end{proof}
\begin{lemma}\label{lem:sesqui} Let $X\subseteq A^*$ be a finite set  and set $\ell =\ell_X= \max\{|x| \mid x\in X\}$. Let
$w= w_1\cdots w_m$, $w_i\in A$, $i=1, \ldots, m$, be a palindrome with $m \geq \ell$. If there exist
$u, v \in (\Pre X)\setminus X$ such that  $|u|=p$, $|v|=q$, $p<q$, and
\begin{equation}\label{eq:piq}
 w_{p+1}\cdots w_mu = w_{q+1}\cdots w_mv, 
 \end{equation}
then
\[ w_1\cdots w_{m-p} = \alpha^k \alpha', \]
where $\alpha' \in \Pre \alpha$, $\alpha^{\sim}$ is a prefix of $v$ of length $q- p$, and 
$k \geq \frac{m}{\ell -1}-1$.
\end{lemma}
\begin{proof} Let $u= a_1\cdots a_p$ and $v= b_1\cdots b_q$ with $a_i, b_j\in A$, $i=1, \ldots, p$, $j=1, \ldots, q$. From (\ref{eq:piq}) one derives:  $a_i= b_{q-p+i}$, $i=1, \ldots, p$, and
\[w_{p+1}\cdots w_q(w_{q+1}\cdots w_m)= (w_{q+1}\cdots w_m)b_1\cdots b_{q-p}.\]
From a classic result of Lyndon and Sch\"utzenberger (cf. \cite{LO}), there exist $\lambda, \mu \in A^*$ and
an integer $h\geq 0$ such that:
\begin{equation}\label{eq:piq1}
 w_{p+1}\cdots w_q = \lambda\mu, \ \ b_1\cdots b_{q-p}= \mu\lambda, \ \ w_{q+1}\cdots w_m= (\lambda\mu)^h\lambda.
\end{equation}
Hence,
\[ w_{p+1}\cdots w_m = (\lambda\mu)^{h+1}\lambda.\]
Since $w\in PAL$, one has for any $i=1, \ldots, m$, $w_i= w_{m-i+1}$. Hence, by taking the reversals of both the sides of the preceding equation,  one has:
\[w_1\cdots w_{m-p}= w_m\cdots w_{p+1} = (\lambda^{\sim}\mu^{\sim})^{h+1}\lambda^{\sim}= \alpha^k \alpha', \]
having set $k=h+1$, $\alpha= \lambda^{\sim}\mu^{\sim}$, and $\alpha'= \lambda^{\sim}$. Now from (\ref{eq:piq1}),
$\alpha^{\sim}= \mu\lambda = b_1\cdots b_{q-p}$ is a prefix of $v$.

From (\ref{eq:piq1}) one has that $m-q= h(q-p)+ |\lambda|$. Since $|\lambda|\leq q-p$ it follows that
$m-q\leq h(q-p)+(q-p)= (h+1)(q-p)=k(q-p)$. Hence, $k \geq \frac{m-q}{q-p}$. As $q-p \leq \ell-1$ and
$\frac{q}{\ell -1} \leq 1$, the result follows.
\end{proof}

\begin{lemma}\label{lem:propre} Let $X$ be a finite maximal prefix code over a $d$-letter alphabet.  Then
\[ \card( (\Pre X)\setminus X) = \frac {\card(X)-1}{d-1}.\]
\end{lemma}
\begin{proof} The code $X$ is represented by the set of leaves of a full $d$-ary tree. The elements
of the set $(\Pre X)\setminus X$, i.e.,   the proper prefixes of the words of $X$ are represented by the internal nodes of the tree. As is well known, the number of internal nodes of a full $d$-ary tree is equal to the number of leaves minus 1 divided by $d-1$.
\end{proof}

In the following we let $\lambda_X$ be the quantity 
\[ \lambda_X =  \frac {\card(X)-1}{d-1}.\]

\begin{prop}\label{prop:exp} Let $s$ be an $X$-AR word. There exists an integer  $e_s$ such that
for any non-empty proper prefix $u$ of a word of $X$, one has $u^{e_s} \not \in \Ff s$. Moreover, also $(u^{\sim})^{e_s} \not\in \Ff s$.
\begin{proof} Any word $x\in X$, as well as any prefix of $x$, is a factor of $s$. Let $u$ be any proper non-empty prefix of a word of $X$. From Lemma \ref{lem:uromega} there exists an integer  $p$ such that 
$u^{p}\not \in \Ff s$. Let $e_u$ be the smallest $p$ such that this latter condition is satisfied. Let us
set \[e_s= \max\{ e_v \mid v\in (\Pre X)\setminus (X \cup \{\varepsilon\})\}.\]
We  observe that $e_s$ is finite since $X$ is a finite code. Therefore, for any  $u\in (\Pre X)\setminus  (X \cup \{\varepsilon\})$ one has  
\[ u^{e_s} \not\in \Ff s.\]
Since $s$ is closed under reversal it follows that also $(u^{\sim})^{e_s} \not\in \Ff s$.
\end{proof}
\end{prop}
\begin{thm}Let $s= \psi_X(y)$, with $y= y_1\cdots y_n \cdots \in X^{\omega}$, $y_i\in X$, $i\geq 1$, be an $X$-AR word. There exists an integer $\nu$ such that for all $h\geq \nu$
the number $S_r(h)$ of right special factors of $s$ of length $h$ has the lower bound $\lambda_X$, i.e.,
\[ S_r(h) \geq \lambda_X.\]
Moreover, any such right special factor of $s$ is of degree $d$.
\begin{proof} In the following we shall set for all $n$, $u_n= \psi_X(y_1\cdots y_n)$.  Let $\ell$ be as in Lemma~\ref{lem:sesqui}, $m_0$ be the minimal integer such that  $ \frac{m_0}{\ell -1} -1 \geq e_s$, and let
$n$ be an integer such that $|u_n|=m \geq m_0$. Let us write $u_n$ as $u_n =w_1\cdots w_m$ with $w_i\in A$, $ i=1, \ldots, m$. Since any word $x\in X$ is persistent in $y$ it follows
that $u_nX \subseteq \Ff s$. Therefore, for any proper prefix $u$ of a word $x\in X$ one has that:
$ u_nu = w_1\cdots w_m u$
is a right special factor of $s$ of order $d$ and length $m+|u|$. This implies that
\begin{equation}\label{eq:emmeu}
w_{|u|+1}\cdots w_mu
\end{equation}
is a right special factor of length $m$. However, for  $u,v \in (\Pre X)\setminus X$, $u\neq v$, one cannot have
\[w_{|u|+1}\cdots w_mu = w_{|v|+1}\cdots w_mv.\]
This is trivial if $|u|=|v|$. If $|u|<|v|$, as $u_n \in PAL$, by Lemma \ref{lem:sesqui} one would derive:
\[w_1\cdots w_{m-|u|} = \alpha^k\alpha'\]
with $k \geq e_s$ and $\alpha$ equal to the reversal of a proper prefix of a word of $X$, which is absurd in view of Proposition \ref{prop:exp}.  Thus one has that all the words of (\ref{eq:emmeu}) with $u\in (\Pre X)\setminus X$, are right special factors of $s$ of length $m$ and order $d$. Since by Lemma \ref{lem:propre} the number of proper prefixes of the words of $X$ is $\lambda_X$ it follows that  the number $S_r(m)$ of right special factors of length $m$ has the lower bound $S_r(m) \geq \lambda_X$.
Thus we have proved the result for all $m=|u_n|\geq m_0$.

Let us now take  $h$ such that   $m <h< m'= |u_{n+1}|$.  We can write 
$u_{n+1} = \zeta w_1\cdots w_m$ for some word $\zeta$. 
Since for any $u\in (\Pre X)\setminus X$, $u_{n+1}u$ is a right special factor of $s$ of length $m'+ |u|$ and order $d$, so is its suffix of length $h$. 
We wish to prove that
all such suffixes of length $h$, for different values of $u$ in  $(\Pre X)\setminus X$, are distinct. Indeed, if two such suffixes were equal,
for instance the ones corresponding to $u,v\in (\Pre X)\setminus X$, then their suffixes of length 
 $m$ would be equal, i.e.,
\[w_{|u|+1}\cdots w_m u = w_{|v|+1}\cdots w_m v,\]
which is absurd as shown above. Hence, $S_r(h) \geq \lambda_X.$
\end{proof}
\end{thm}

\begin{cor} Let $s$ be an $X$-AR word. There exists an integer $\nu$ such that 
the factor complexity $p_s$ of $s$ has for all $n\geq \nu$ the linear lower bound
\[(\card(X)-1)n +c, \  \mbox{with} \  c\in \Int. \]
\end{cor}
\begin{proof} From the preceding theorem for all $n\geq \nu$, $s$ has at least  $\lambda_X$ right special factors of length $n$ and order $d$. Therefore, in view of (\ref{eq:rsf}), we can write for all $n\geq \nu$
\[p_s(n) \geq p_s(\nu)+ (n-\nu)\lambda_X (d-1)= p_s(\nu)+ (n-\nu)(\card(X)-1)\]\[= (\card(X)-1)n+c,\]
having set $c= p_s(\nu)-\nu(\card(X)-1)$.
\end{proof}
We shall prove that the factor complexity $p_s$ of an $X$-AR word $s$ is linearly upper bounded
(cf. Theorem \ref{thm:upperbound}). We need some preparatory results and a theorem (cf. Theorem \ref{thm:seed}) which is a suitable extension of a formula of Justin \cite{J} to generalized palindromization maps.

We recall that a positive integer  $p$ is a {\em period} of the word $w=a_1\cdots a_n$, $a_i\in A$,
$1\leq i\leq n$ if the following condition is
satisfied: if  $i$ and $j$ are any integers such that $1\leq i,j\leq n$
and $i\equiv j \pmod{p}$, then $ a_i = a_j$. We shall denote by $\pi(w)$ the minimal period of $w$. 

Let $X$ be a finite prefix code and $\ell_X$ be the maximal length of the words of $X$. We say that
$\psi_X(x_1\cdots x_m)$ with $x_i\in X$, $i\geq 1$, is {\em full} if it satisfies  the three following conditions:
\begin{enumerate}
\item[F1.] For any $x\in X$ there exists at least one integer  $j$ such that  $1\leq j \leq m$ and $x_j=x$.
\item[F2.] $\pi(\psi_X(x_1\cdots x_m)) \geq \ell_X$.
\item[F3.] For all $x\in X$ the longest palindromic prefix of $\psi_X(x_1\cdots x_m)$ followed by $x$
is $\psi_X(x_1\cdots x_{r_x-1})$, where $r_x$ is the greatest integer such that  $1\leq r_x\leq m$ and $x_{r_x}=x$.
\end{enumerate}

\begin{prop}\label{prop:full000} Let $X$ be a finite prefix code, $z\in X^{+}$, and $y\in X$.
If $\psi_X(z)$ is full, then $\psi_X(zy)$ is full.
\end{prop}
\begin{proof} 
It is clear that $\psi_X(zy)$ satisfies
property F1. Moreover, one has also that $\pi(\psi_X(zy))\geq \ell_X$. Indeed, otherwise
since $\psi_X(z)$ is a prefix of $\psi_X(zy)$, one would derive that $\psi_X(z)$ has a period, and then the minimal period, less than $\ell_X$, which is a contradiction.

Let us first prove that $\psi_X(z)=P $, where $P$ is the longest proper palindromic prefix of $\psi_X(zy)$. Indeed, we can write:
$$ \psi_X(zy) = \psi_X(z)y\lambda = P \mu,$$
with $\lambda, \mu\in A^*$ and $\mu \neq \varepsilon$. One has that $|P|\geq |\psi_X(z)|$ and, moreover,
$|P|<|\psi_X(z)y|$. This last inequality follows from the minimality of the length of palindromic closure.
Let us then suppose that:
$$ P = \psi_X(z)y'= (y')^{\sim}\psi_X(z),$$
with $y' \prec y$. From the Lyndon and Sch\"utzenberger theorem there exist $\alpha,\beta\in A^*$
and $n\in \Nn$ such that  $(y')^{\sim} = \alpha\beta, y'= \beta\alpha$, and $\psi_X(z)= (\alpha\beta)^n\alpha$. Since $\psi_X(z)$ is full, from property F1 one has that $|\psi_X(z)|\geq \ell_X$, so that $n>0$
and $\pi(\psi_X(z))\leq|\alpha\beta|=|y'|<\ell_X$ which is a contradiction. Thus $P=\psi_X(z)$.

From the preceding result one derives that the longest palindromic prefix of $\psi_X(zy)$ followed by
$y$ is $\psi_X(z)$. Now let $x\neq y$ and let $Q$ be the longest palindromic prefix of $\psi_X(zy)$
followed by $x$. We can write:
$$ \psi_X(zy) = \psi_X(z)y\lambda= Qx\delta,$$
with $\delta\in A^*$. From the preceding result one has $|Q|\leq |\psi_X(z)|$. If $|Q|= |\psi_X(z)|$, then, as $X$ is a prefix code, one gets $x=y$, a contradiction. Hence, $|Q|< |\psi_X(z)|$. We have to consider two cases:

\noindent
Case 1. $|Qx|>|\psi_X(z)|$. This implies 
$$ \psi_X(z) = Q x'= (x')^{\sim}Q,$$
with $x' \prec x$. Hence, one would derive $(x')^{\sim}= uv, x'=vu$, and $\psi_X(z)= (uv)^nu$ with
$u,v\in A^*$ and $n>0$. This gives rise to a contradiction, as $\pi(\psi_X(z))\leq |uv|<\ell_X$.

\noindent
Case 2. $|Qx|\leq |\psi_X(z)|$. Let $z= x_1\cdots x_m$ with $x_i \in X$, $1\leq i \leq m$. In this case $Q$ is the longest palindromic prefix of $\psi_X(z)$ followed by $x$, namely $\psi_X(x_1\cdots x_{r_x-1})$.

In conclusion,  $\psi_X(zy)$ satisfies  conditions F1--F3 and is then full.
\end{proof}

\begin{lemma}\label{lemma:arnx0} Let  $s$ be an $X$-AR word and  $ \psi_X(z)$, with $z\in X^*$,
 be a  prefix of  $s$. There exists an integer $\nu_s$ such that if $|\psi_X(z)| \geq \nu_s$, 
then for any prefix  $u= \psi_X(zyx_1\cdots x_k)$ of $s$ with $k\geq 0$, $y, x_1,\ldots, x_k\in X$,   
$y\neq x_i$, $1\leq i\leq k$,
the longest palindromic prefix of  $u$ followed by $y$ is $\psi_X(z)$.  
\end{lemma}
\begin{proof} Let us denote by $P$ the longest palindrome such that $Py$ is a  prefix of $u$. We wish to prove that for a sufficiently large $\psi_X(z)$ one has that $P= \psi_X(z)$. Let us then suppose by contradiction  that $|P|> |\psi_X(z)|$. Setting $x_0=y$, 
 there exists
an integer  $i$,   $-1\leq i \leq k-1$ such that
\begin{equation}\label{eq:arra0}
|\psi_X(zx_0\cdots x_i)| \leq |P| \leq |\psi_X(zx_0\cdots x_{i+1})|,
\end{equation}
where for $i=-1$ the  l.h.s.~of the preceding equation reduces to $|\psi_X(z)|$. Let us prove that for
$-1 \leq i \leq k$, $P \neq \psi_X(zx_0\cdots x_i)$. This is trivial for  $i=-1$ and $i=k$ as $|P|< |u|$. For
$0\leq i\leq k-1$ the result is a consequence of the fact that $P$ is followed by $y$ whereas $\psi_X(zx_0\cdots x_i)$ is followed by $x_{i+1}$. As $X$ is a prefix code, one would obtain $y=x_{i+1}$
which is a contradiction.  Hence in $(\ref{eq:arra0})$ the inequalities are strict. If  
\[ |\psi_X(zx_0\cdots x_i)x_{i+1}| \leq |P| < |\psi_X(zx_0\cdots x_{i+1})|,\]
then one would contradict the definition of palindromic closure. Thus the only possibility is that
there exists  $-1\leq i \leq k-1$ such that
\[ P = \psi_X(zx_0\cdots x_i)p = p^{\sim} \psi_X(zx_0\cdots x_i)\]
where $p$ is a proper non-empty prefix of  $x_{i+1}$. This implies that there exist words $\lambda, \mu\in A^*$ and an integer $n\geq 0$ such that
\begin{equation}\label{eq:arra1}
 p^{\sim} = \lambda\mu, \  p= \mu\lambda, \ \psi_X(zx_0\cdots x_i)= (\lambda\mu)^n\lambda.
\end{equation}
Let us set  $\nu_s = (e_s+1)\ell_X$, where $e_s$ has been defined in Proposition \ref{prop:exp}
and $\ell_X$ is the maximal length of the words of $X$.
Let us suppose that  $|\psi_X(z)|\geq \nu_s$.  Since  \[(e_s+1)\ell_{X} \leq |\psi_X(z)| \leq |\psi_X(zx_0\cdots x_i)| \leq (n+1)\ell_X,\]
one would derive  $n\geq e_s$ and $p^{n}\not\in \Ff s$ which  contradicts (\ref{eq:arra1}) and
this concludes the proof.
\end{proof}

\begin{cor} 
\label{thm:full}
Let $s=\psi_{X}(x_{1}x_{2}\cdots x_{n}\cdots)$ be an $X$-AR word, with $x_{i}\in X$, $i\geq 1$. There exists an integer $m\geq 1$ such that for all $n\geq m$,
$\psi_{X}(x_{1}\cdots x_{n})$ is full.
\end{cor}
\begin{proof} Since $s$ is an $X$-AR word, for any $x\in X$
there exist infinitely many integers $j$ such that $x=x_j$. We can take the integer $m$ so large  that for any $x\in X$ there exists at least one integer $j$ such that $1\leq j\leq m$, $x_j=x$, and, moreover, for each $x\in X$
 \[ |\psi_X(x_1\cdots x_{r_x-1})| > \nu_s. \]
This assures, in view of preceding lemma, that for each $x\in X$ the longest palindromic prefix of $\psi_X(x_1\cdots x_m)$ followed by $x$ is $\psi_X(x_1\cdots x_{r_x-1})$. Finally,  there exists an integer $m$ such that $\pi(\psi_X(x_1\cdots x_m))\geq \ell_X$.
Indeed, $s$ is 
$\omega$-power free, so that there exists an integer $p$ such that for any non-empty factor $u$ of $s$ of length  $|u|<\ell_X$
one has $u^p\not\in \Ff s$. Thus if for all $m$,  $\pi(\psi_X(x_1\cdots x_m))< \ell_X$ we reach a
contradiction by taking $m$ such that  $|\psi_X(x_1\cdots x_m)|\geq (p+1)\ell_X$.
Hence there exists an integer $m$ such that conditions F1--F3 are all satisfied, so that $\psi_{X}(x_{1}\cdots x_{m})$ is full. By  Proposition \ref{prop:full000}, 
$\psi_{X}(x_{1}\cdots x_{n})$ is also full, for all $n\geq m$.
\end{proof}

\begin{lemma}
\label{thm:psi-1}
Let $z\in X^{*}$ and $y\in X$. Suppose
that $\psi_{X}(z)$ has some palindromic prefixes followed by $y$, and let $\Delta_y$ be the longest one. Then
\[\psi_{X}(zy)=\psi_{X}(z)\Delta_y^{-1}\psi_{X}(z)\;.\]
\end{lemma}
\begin{proof}
Since $\Delta=\Delta_y$ is the longest palindromic prefix of $\psi_{X}(z)$ followed by $y$, it is also the longest palindromic suffix preceded by $y^{\sim}$, so that $ y^{\sim}\Delta y$ is the longest palindromic suffix of $\psi_{X}(z)y$. Thus, letting
$\psi_{X}(z)=\Delta y\zeta= \zeta^{\sim} y^{\sim}\Delta$ for a suitable $\zeta$, we obtain
\[\psi_{X}(zy)=\left(\psi_{X}(z)y\right)^{(+)}= \zeta^{\sim} y^{\sim}\Delta y\zeta
=\psi_{X}(z)\Delta^{-1}\psi_{X}(z)\;. \qedhere\]
\end{proof}

Let $B$ be a finite alphabet and $\mu: B\rightarrow X$ be a bijection to a 
prefix code $X\subseteq A^*$. For $z\in X^*$, we define a morphism
$\varphi_{z}: B^*\rightarrow A^*$ by setting for all $b\in B$
\begin{equation}\label{eq:jusgen1}
\varphi_{z}(b)=\psi_{X}(z\mu(b))\psi_{X}(z)^{-1}= \psi_X(z)\Delta_{\mu(b)}^{-1},
\end{equation}
where for the last equality we used Lemma \ref{thm:psi-1}.

\begin{thm}
\label{thm:seed}
Let $s=\psi_{X}(x_{1}x_{2}\cdots x_{n}\cdots)$ be an X-AR word with $x_i\in X$, $i\geq 1$.     If $z=x_{1}\cdots x_{m}$ is 
 such that $u_{m}=\psi_{X}(z)$ is full and
    $\mu$, $\varphi_z$ are defined as
    above, then for any $w \in B^*$ the following holds:
    \[\psi_X(z\mu(w))=\varphi_z(\psi(w))\psi_X(z)\;.\]
\end{thm}

\begin{proof}
    In the following we shall  use the readily verified property that if $\gamma: B^*\rightarrow A^*$
    is a morphism and $v$ is a suffix of $u\in B^*$,
    then $\gamma(uv^{-1})=\gamma(u)\gamma(v)^{-1}$.
    
    We will prove the theorem by induction on $|w|$. It is 
    trivial that for $w=\varepsilon$ the claim is true since 
    $\psi(\varepsilon)=\varepsilon=\varphi_z(\varepsilon)$.
    Suppose that for all the words shorter than $w$, the statement holds.
    For $|w| > 0$, we set $w=vb$ with $b\in B$, and let $y=\mu(b)$. 
    
    First we consider the case $|v|_{b}\neq 0$. We can then write $v=v_{1}bv_{2}$ 
    with $|v_{2}|_{b}=0$.
    Since $\psi_{X}(z)$ is full, so is $\psi_{X}(z\mu(v))$; hence
    $\psi_{X}(z\mu(v_{1}))$ is the longest 
    palindromic prefix (resp.~suffix) followed (resp.~preceded) by $y$ (resp.~$ y^{\sim}$) in $\psi_{X}(z\mu(v))$.
    Therefore, by Lemma~\ref{thm:psi-1} we have
\begin{equation}\label{caso1}
    \psi_{X}(z\mu(v)y)=
  \psi_{X}(z\mu(v))\psi_{X}(z\mu(v_{1}))^{-1}\psi_{X}(z\mu(v))
\end{equation}
and, as $\psi(v_1)$ is the longest palindromic prefix (resp. suffix) followed (resp. preceded) by $b$ in $\psi(v)$,
\begin{equation}
\label{eq:psi}
\psi(vb)=\psi(v)\psi(v_1)^{-1}\psi(v)\;.
\end{equation}
By induction we have:
\[\psi_{X}(z\mu(v))=\varphi_z(\psi(v))\psi_{X}(z)\;,\quad
	 \psi_{X}(z\mu(v_{1}))=\varphi_z(\psi(v_{1}))\psi_{X}(z)\;.
\]
Replacing in (\ref{caso1}), and by (\ref{eq:psi}), we obtain 
    \begin{eqnarray*}    \psi_{X}(z\mu(v)y)&=&\varphi_z(\psi(v))\varphi_z(\psi(v_1))^{-1}\varphi_z(\psi(v))\psi_X(z)\\
    &=&\varphi_z(\psi(v)\psi(v_1)^{-1}\psi(v))\psi_X(z)\\
    &=&\varphi_z(\psi(vb))\psi_{X}(z)\;, 
    \end{eqnarray*}
    which was our aim.

Now suppose that $|v|_b=0$. As $\psi_{X}(z)$ is full, the longest palindromic prefix of $\psi_{X}(z)$ which is followed by $y$ is
$\Delta_y=\psi_{X}(x_{1}\cdots x_{r_{y}-1})$, where $r_{y}$ is the greatest integer such that $1\leq r_y\leq m$ and $x_{r_y}=y$. By Lemma~\ref{thm:psi-1} we obtain
\begin{equation}
\label{eq:incaso2}
\psi_X(z\mu(v)y)=\left(\psi_X(z\mu(v))y\right)^{(+)}=\psi_X(z\mu(v))\Delta_y^{-1}\psi_X(z\mu(v))\;.
\end{equation}
By induction, this implies
\begin{equation}
\label{eq:caso2ind}
\psi_X(z\mu(v)y)=\varphi_z(\psi(v))\psi_X(z)\Delta_y^{-1}\varphi_z(\psi(v))\psi_X(z)\;.
\end{equation}
From (\ref{eq:jusgen1}) it follows
\[\varphi_z(b)=\psi_X(zy)\left(\psi_X(z)\right)^{-1}=\psi_X(z)\Delta_y^{-1}\;.\]

Moreover, since $\psi(v)$ has no palindromic prefix (resp.~suffix) followed (resp.~preceded) by $y$ one has
\begin{equation}
\label{eq:psi2}
\psi(vb)=\psi(v)b\psi(v)\;.
\end{equation}
Thus from (\ref{eq:caso2ind}) we obtain
\begin{eqnarray*}
\psi_X(z\mu(v)y)&=&\varphi_z(\psi(v))\varphi_z(b)\varphi_z(\psi(v))\psi_X(z)\\
&=&\varphi_z(\psi(v)b\psi(v))\psi_X(z)\\
&=&\varphi_z(\psi(vb))\psi_X(z)\;,
\end{eqnarray*}
    which completes the proof.
\end{proof}

\begin{cor}
Every $X$-AR word is a morphic image of a standard Arnoux-Rauzy word over an alphabet $B$ of the same cardinality as $X$.
\end{cor}
\begin{proof}
Let $s=\psi_{X}(x_{1}x_{2}\cdots x_{n}\cdots)$ be an $X$-AR word with $x_i\in X$, $i\geq 1$, and let 
$x_{i}=\mu(b_{i})$ for all $i\geq 1$, where $\mu:B\to X$ is a bijection. By the preceding theorem, there exists an integer $m\geq 1$ such that, setting $z=x_{1}\cdots x_{m}$, for all $w\in B^{*}$ we have $\psi_{X}(z\mu(w))=\varphi_{z}(\psi(w))\psi_{X}(z)$. Hence for all $k\geq m$ we have
\[\psi_{X}(x_{1}\cdots x_{k})=
\varphi_{z}(\psi(b_{m+1}\cdots b_{k}))\psi_{X}(z)\;,\]
so that taking the limit of both sides as $k\to\infty$, we get
\[s=\varphi_{z}(\psi(b_{m+1}b_{m+2}\cdots b_{n}\cdots))\;.\]
The assertion follows, as each letter of $B$ occurs infinitely often in the word
$b_{m+1}b_{m+2}\cdots b_n \cdots$.
\end{proof}

\begin{example}
Let $X=\{aa,ab,b\}$, $B=\{a,b,c\}$, and
$\mu:B\to X$ be defined by $\mu(a)=ab$, $\mu(b)=b$, and $\mu(c)=aa$.
Let $s$ be the $X$-AR word
\[s=\psi_{X}((abbaa)^{\omega})=ababaaababaababaaabababaaababaababaaaba\cdots\;.\]
Setting $z=abbaa$, it is easy to verify that the prefix $\psi_{X}(z)=ababaaababa$ of $s$ is full, so that $s=\varphi_{z}\left(\psi\left((abc)^{\omega}\right)\right)$, where
$\varphi_{z}(a)=ababaaababa$, $\varphi_{z}(b)=ababaaab$, and
$\varphi_{z}(c)=ababaa$.
\end{example}

Let $s= \psi_X(x_1\cdots x_n\cdots)$ be an $X$-AR word with $x_i\in X$, $i\geq 1$, and let $m_0$ be the minimal integer
such that  $u_{m_0}=\psi_X(x_1\cdots x_{m_0})$ is full. For all $j\geq 0$ we shall set
\[ \alpha_j= u_{m_0+j} \ \mbox{and} \ \ n_j= |\alpha_j|.\]
\begin{thm}\label{thm:upperbound} Let $s$ be an $X$-AR word. Then  the factor complexity of
$s$ is linearly upper bounded. More precisely for all $n\geq n_0$
\[ p_s(n) \leq 2\card(X) n -\card(X).\]
\end{thm}
\begin{proof} We shall first prove that for all $j\geq 0$
\begin{equation}\label{eq:piesse1}
 p_s(n_j) \leq \card(X)n_j -\card(X).
\end{equation}
Let $\mu$ be a bijection of an alphabet $B$ and $X$. We set  $z_j= x_1\cdots x_{m_0+j}$
and consider the morphism  $\varphi_{z_j}: B^* \rightarrow A^*$ defined, in view of (\ref{eq:jusgen1}),   for all $b\in B$ as:
\[\varphi_{z_j}(b) = \alpha_j \Delta_{\mu(b)}^{-1},\]
where $\alpha_j= \psi_X(z_j)$ and $\Delta_{\mu(b)}$ is the longest palindrome such that $\Delta_{\mu(b)}\mu(b)$ (resp. $(\mu(b))^{\sim}\Delta_{\mu(b)}$) is a prefix (resp. suffix) of  $\alpha_j$. 

Since $s$ is uniformly recurrent, there exists an integer $p$ such that all factors of $s$ of length $n_j$
are factors of $\alpha_{j+p}$.
Hence, there exist $p$ letters $b_1,\ldots, b_p \in B$ such that $\alpha_{j+p}= \psi_X(z_j\mu(b_1)\cdots \mu(b_p))$. 
 By Theorem \ref{thm:seed} one has
\[\psi_X(z_j\mu(b_1)\cdots \mu(b_p)) = \alpha_j\Delta_{\mu(b_1)}^{-1} \alpha_j\Delta_{\mu(b_2)}^{-1}\cdots \alpha_j\Delta_{\mu(b_p)}^{-1}\alpha_j.\]
Thus $\alpha_j$ covers $\alpha_{j+p}$  and the overlaps between two consecutive occurrences
of $\alpha_j$ in $\alpha_{j+p}$ are given by  $\Delta_{\mu(b_i)}$, $1\leq i \leq p$. Any factor of $s$ of length $n_j$ will be a factor of two consecutive overlapping occurrences of $\alpha_j$, i.e., of
\begin{equation}\label{eq:overl1}
 \alpha_j\Delta_{\mu(b_i)}^{-1} \alpha_j, \ i=1,\ldots, p. 
\end{equation}
For any $1\leq i \leq p$ the number of distinct factors in (\ref{eq:overl1}) is at most  $n_j-|\Delta_{\mu(b_i)}|\leq n_j-1$. Since $\mu(B)= X$  and the number of distinct consecutive overlapping occurrences of $\alpha_j$ in $\alpha_{j+p}$ is at most $\card(X)$, equation (\ref{eq:piesse1}) is readily derived.

Now let $n$ be any integer $n\geq n_0$ such that $n\neq n_k$ for all $k\geq 0$.  There exists an integer $j$ such that  $n_j< n <   n_{j+1}$. Since $s$ is not periodic, by a classic result of Morse and Hedlund  (see \cite[Theorem~1.3.13]{LO2}) the factor complexity $p_s$ is strictly increasing with $n$. Moreover,  as $n_{j+1}<2n_j<2n$, one has by (\ref{eq:piesse1}):
\[p_s(n)<p_s(n_{j+1})\leq  \card(X)n_{j+1}-\card(X)<2\card(X)n -\card(X),\]
which concludes the proof.
\end{proof}

\section{Conservative  maps}\label{sec:six}

Let $A$ be an alphabet of cardinality $d>1$ and let $X$ be a code over $A$. We say that the palindromization map $\psi_X$ is \emph{conservative}
if 
\begin{equation}\label{eq:con1}
 \psi_X(X^*) \subseteq X^*.
\end{equation}
When  $X=A$,  the palindromization map $\psi$ is trivially always conservative. In the general case $\psi_X$ may be non conservative.

\begin{example}\label{ex:cons} Let $X=\{ab, ba\}$. One has $\psi_X(ab)= aba \not \in X^*$, so that $\psi_X$ is not conservative. In the case $Y=\{aa, bb\}$ one easily verifies that  $\psi_Y(Y^*) \subseteq Y^*$. If
$Z=\{a, ab\}$ one has that for any word $w\in Z^*$, $\psi_Z(w) \in aA^*\setminus A^*bbA^*$, with $A=\{a, b\}$, so that it can be uniquely factorized by the elements of $Z$. This implies that $\psi_Z$ is conservative.
\end{example}

The following result shows that a prefix code having a conservative palindromization map allows a natural generalization of properties P2 and P3 of Proposition~\ref{prop:basicp}, in addition to the ones for P1 and P4 shown in Propositions~\ref{prop:pref} and~\ref{prop:prefcode}.

\begin{prop}
\label{prop:p2e3}
Let $X$ be a prefix code such that $\psi_{X}$ is conservative, and $p,w\in X^{*}$ with $p$ a prefix of $\psi_{X}(w)$.
The following hold:
\begin{enumerate}
\item $p^{(+)}$ is a prefix of $\psi_{X}(w)$ and $p^{(+)}\in X^*$.
\item If $p$ is a palindrome, then $p=\psi_{X}(u)$ for some prefix $u\in X^*$ of $w$.
\end{enumerate}
\end{prop}
\begin{proof}
Let $w=x_{1}x_{2}\cdots x_{k}$ with $x_{i}\in X$ for $1\leq i\leq k$, and let $v$ be the longest prefix of $w$ in $X^{*}$ such that $\psi_{X}(v)$ is a prefix of $p$; we can write $v=x_{1}\cdots x_{n}$ or set $n=0$ if $v=\varepsilon$. Thus $p= \psi_X(v)\zeta$ with $\zeta\in A^*$. Since $\psi_X$ is conservative one has
$\psi_X(v)\in X^*$. Moreover, as $X$ is a prefix code, $X^*$ is right unitary, so that one has  $\zeta\in X^*$. If $\zeta= \varepsilon$, then $p=\psi_{X}(v)=p^{(+)}$ and there is nothing to prove. Let us then suppose $\zeta \neq \varepsilon$. Since $\psi_X(v)x_{n+1}$, as well as $p$,  is a prefix of $\psi_X(w)$
and $X$ is a prefix code, one has that $\zeta \in x_{n+1}X^*$. Thus
 $\psi_{X}(v)x_{n+1}$ is a  prefix of $p$. 
 
 From the definition of palindromic closure it follows that  $ |(\psi_X(v)x_{n+1})^{(+)}|$ $\leq$ $|p^{(+)}|$. By the maximality of $n$, we also obtain that $p$ is a (proper) prefix of $(\psi_{X}(v)x_{n+1})^{(+)}=\psi_{X}(x_{1}\cdots x_{n+1})$, so that $|p^{(+)}| \leq |(\psi_X(v)x_{n+1})^{(+)}|$.
 Thus $|p^{(+)}| =|(\psi_X(v)x_{n+1})^{(+)}|$. Since $p^{(+)}$ is a palindrome of minimal length having $\psi_X(v)x_{n+1}$ as a prefix,  from the uniqueness of palindromic closure it follows that $p^{(+)}=\psi_X(vx_{n+1})$. Hence, $p^{(+)}$ is a prefix of $\psi_X(w)$, and $p^{(+)}\in X^*$ as $\psi_X$
 is conservative.
 
If $p$ is a palindrome and $p\neq\psi_{X}(v)$, then the argument above shows that $p^{(+)}=p=\psi_{X}(vx_{n+1})$, which is absurd by the maximality of $n$.
\end{proof}
The following proposition gives a sufficient condition which assures that $\psi_X$ is conservative.

\begin{prop} \label{prop:suff}Let $X\subseteq PAL$ be an  infix and weakly overlap-free code.
Then $\psi_X$ is conservative.
\end{prop}
\begin{proof} We shall prove that for all $n\geq 0$ one has that $\psi_X(X^n)\subseteq X^*$.  The proof
is by induction on the integer $n$. The base of the induction is true.  Indeed the case $n=0$ is trivial and for $n=1$, since $X\subseteq PAL$,  one has $\psi_X(X)= X$.
Let us then suppose the result true up to $n$ and prove it for $n+1$. Let $w\in X^n$ and $x\in X$.
By induction we can write $\psi_X(w)= x'_1\cdots x'_m$, with $x'_i\in X$, $1\leq i \leq m$. Thus:
\begin{equation}\label{eq:codeweakly}
 \psi_X(wx) = (\psi_X(w)x)^{(+)} = (x'_1\cdots x'_mx)^{(+)}.
\end{equation}
Let $Q$ denote the longest palindromic suffix of  $x'_1\cdots x'_mx$. Since $x\in PAL$ we have
$|Q|\geq |x|$. We have to consider two cases:

\vspace{2 mm}

\noindent
Case 1. $|Q|= |x|$.  From (\ref{eq:codeweakly}) and $X\subseteq PAL$, it follows:
\[  \psi_X(wx) = (x'_1\cdots x'_mx)^{(+)}= x'_1\cdots x'_m x x'_m\cdots x'_1.\]
Thus $\psi_X(wx)\in X^*$ and in this case we are done.

\vspace{2 mm}

\noindent
Case 2. $|Q| > |x|$. One has:
\[x'_1\cdots x'_mx = \zeta Q.\]
Since $|Q| > |x|$ and $x,Q\in PAL$, there exists  $1\leq j \leq m$ such that $x'_j= \lambda \mu$, $\lambda, \mu\in A^*$
and
\[ \mu x'_{j+1}\cdots x'_m x = Q = x \eta,\]
with $\eta\in A^*$. We shall prove that $\lambda= \varepsilon$. Indeed,  suppose that  $\lambda\neq \varepsilon$. We have to consider the following subcases:

\vspace{2mm}

\noindent
1) $|x|\leq |\mu|$. This implies that $x$ is a proper factor of $x'_j$ which is a contradiction, since
 $X$ is an infix code.

\vspace{2mm}

\noindent
2) $|x|\geq |\mu x'_{j+1}|$. In this case one has that $x'_{j+1}$ is a factor of $x$ which is a contradiction.

\vspace{2mm}

\noindent
3) $|\mu|< |x| < |\mu x'_{j+1}|$. This implies that $x = \mu p$, where $p$ is a proper prefix of $x'_{j+1}$. Since $\mu$ is a proper suffix of $x'_j$ we reach a contradiction with the hypothesis that $X$ is weakly overlap-free.

\vspace{2 mm}

\noindent
Hence, $\lambda= \varepsilon$ and $\mu= x = x'_j$. Therefore, one has, as $X\subseteq PAL$,
\[ \psi_X(wx) = (x'_1\cdots x'_mx)^{(+)}=x'_1\cdots x'_{j-1}x x'_{j+1}\cdots x'_m x x'_{j-1}\cdots x'_1 \in X^*,\]
which concludes the proof.
\end{proof}

\begin{example} Let $X=\{bab, bcb\}$. One has that $X\subseteq PAL$. Moreover, $X$ is an infix and weakly overlap-free code. From the preceding proposition one has that $\psi_X$ is conservative.
\end{example}

Let us observe that  Proposition \ref{prop:suff}  can be proved  by replacing the requirement $X\subseteq PAL$  with the two conditions: $X= X^{\sim}$ and $\psi_X(X)\subseteq X^*$. However, the following lemma shows that if the code $X$ is prefix these two latter conditions are equivalent to $X\subseteq PAL$.

\begin{lemma} Let $X$ be a prefix code. Then one has:
\[ X\subseteq PAL  \Longleftrightarrow X= X^{\sim}  \ \mbox{and} \ \ \psi_X(X)\subseteq X^*.\]
\end{lemma}
\begin{proof} If  $X\subseteq PAL$, then trivially $X=X^{\sim}$. Moreover, for any $x\in X$ one has
$\psi_X(x)= x^{(+)}= x\in X^*$.
Let us prove the converse.  Suppose that $x\in X$ is not a palindrome. We can write $x= \lambda Q$,
where $Q=LPS(x)$ is the longest palindromic suffix of $x$ and $\lambda \neq \varepsilon$. One has, by hypothesis:
\[\psi_X(x) = x^{(+)}= \lambda Q \lambda^{\sim}= x\lambda^{\sim} \in X^*.\]
Since $X$ is a prefix code, from the right unitarity of $X^*$ one has  $\lambda^{\sim}\in X^*$. As $X=X^{\sim}$
it follows $\lambda\in X^*$. Since $x= \lambda Q$ and $X$ is a prefix code, one derives $\lambda=x$ and $Q=\varepsilon$ which is absurd as $|Q|>0$.
\end{proof}

\begin{prop}\label{prop:cpsi} Let $X\subseteq PAL$ be a prefix code. Then:
\[ \psi_X \  \mbox{ is conservative}  \Longleftrightarrow \  \mbox{for all} \ x\in X, \  LPS(\psi_X(X^*)x)\subseteq X^*.\]
\end{prop}
\begin{proof} $(\Rightarrow)$ Let $w\in X^*$. If $w= \varepsilon$, since $X\subseteq PAL$, one has $LPS(x)=x\in X$. Suppose $w\neq \varepsilon$, so that  $w= x_1\cdots x_n$, with $x_i\in X$, $i=1, \ldots, n$. Let $x\in X$ and $Q$ be the longest palindromic suffix of $\psi_X(x_1\cdots x_n)x$. We can write:  $\psi_X(x_1\cdots x_n)x = \delta Q$ with $\delta \in A^*$ and
\[ \psi_X(x_1\cdots x_nx) = (\psi_X(x_1\cdots x_n)x)^{(+)}= \delta Q \delta^{\sim}= \psi_X(x_1\cdots x_n)x\delta^{\sim}.\]
Since $\psi_X$ is conservative, one has $\psi_X(x_1\cdots x_n), \psi_X(x_1\cdots x_nx)\in X^*$,
so that as $X$ is a prefix code from the preceding equation one derives $\delta^{\sim}\in X^*$ and
then  $\delta\in X^*$ because $X\subseteq PAL$. Finally, from the equation $\psi_X(x_1\cdots x_n)x = \delta Q$
 it follows $Q\in X^*$ as $X$ is a prefix code.

\vspace{2 mm}

\noindent
$(\Leftarrow)$ We shall prove that for all $n\geq 0$ one has $\psi_X(X^n)\subseteq X^*$. The result is trivial if $n=0$. For $n=1$ one has that for any $x\in X$, $\psi_X(x)= x^{(+)}= x$ as $X\subseteq PAL$, so that $\psi_X(X)\subseteq X^*$. Let us now by induction suppose that $\psi_X(X^n)\subseteq X^*$ and prove that $\psi_X(X^{n+1})\subseteq X^*$. Let $x_1, \ldots, x_n, x \in X$ and let $Q$ denote the longest palindromic suffix of $\psi_X(x_1\cdots x_n)x$, so that
\[ \psi_X(x_1\cdots x_n)x = \delta Q,\]
with $\delta\in A^*$. The code $X$ is bifix  because $X$ is a prefix code and  $X\subseteq PAL$. Since by hypothesis $Q, \psi_X(x_1\cdots x_n) \in X^*$, from the preceding equation  and the left unitarity of $X^*$, one gets $\delta \in X^*$. 
Moreover, $\delta^{\sim}\in X^*$ since $X\subseteq PAL$. Hence, one has:
\[\psi_X(x_1\cdots x_nx) = (\psi_X(x_1\cdots x_n)x)^{(+)}= \delta Q \delta^{\sim}\in X^*,\]
which concludes the proof.
\end{proof} 

Let $X$ be a code over the alphabet $B$ and $\varphi: A^* \rightarrow B^*$ an injective morphism such that $\varphi(A)= X$. We say that $\psi_X$ is  \emph{morphic-conservative} if for any $w\in A^*$ one has
\begin{equation}\label{eq:con2}
 \varphi(\psi(w)) = \psi_X(\varphi(w)).
\end{equation}

\begin{example} Let $A=\{a,b\}$, $B=\{a,b,c\}$, $X=\{c, bab\}$, and $\varphi: A^*\rightarrow B^*$ be  the injective
morphism defined by $\varphi(a)= c$ and $\varphi(b)= bab$. Let $w= abaa$; one has
$\psi(w) = abaabaaba$, $\varphi(w)= cbabcc$, and
\[ \varphi(\psi(w)) = cbabccbabccbabc = \psi_X(\varphi(w)).\]
As a consequence of Corollary \ref{cor:morfconsA}, one can prove that $\psi_X$ is morphic-conserva\-tive.
\end{example}

\begin{lemma} \label{lemma:morcon00}If $\psi_X$ is morphic-conservative, then it is conservative.
\end{lemma}
\begin{proof} Let $u\in X^*$. The result is trivial if $u=\varepsilon$. If $u$ is not empty let us write
$u=x_1\cdots x_n$, with $x_i\in X$, $i=1,\ldots, n$. Since $\varphi$ is injective, let $a_i\in A$ be the unique letter such that $x_i= \varphi(a_i)$. Therefore, $u= \varphi (a_1\cdots a_n)$. By (\ref{eq:con2}) one has  $\psi_X(u)= \varphi(\psi(a_1\cdots a_n))\in X^*$, which proves the assertion.
\end{proof}

The converse of the preceding lemma is not true in general. Indeed, from the following proposition,  one has that if $\psi_X$ is morphic-conservative, then the words of $X$ have to be palindromes. However, as we have seen in Example \ref{ex:cons}, there are $\psi_X$ which are conservative with a code $X$ whose words are not palindromes.

\begin{prop}\label{prop:bifixpal} If $\psi_X$ is morphic-conservative, then $X\subseteq PAL$ and $X$ has to be a bifix code.
\end{prop}
\begin{proof} Let $a$ be any letter of $ A$ and set $x=\varphi(a)$.  One has from (\ref{eq:con2}) that $\varphi(\psi(a))= \varphi(a)= x= \psi_X(\varphi(a)) = \psi_X(x) = x^{(+)}$. Hence, $x= x^{(+)}\in PAL$, so that all the words of $X$ have to be palindromes.

Let us now prove that $X$ is a suffix code. Indeed,
 suppose by contradiction that there exist words
$x,y \in X$ such that  $y = \lambda x$ with $\lambda\in A^+$. Let $a,b\in A$ be letters such that
$\varphi(a)=x$ and $\varphi(b)=y$. For  $w=ba$ one has:
\[\varphi(\psi(ba)) = \varphi(bab) = yxy,\]
and, recalling that $y\in PAL$,
\[\psi_X(\varphi(ba))= \psi_X(yx)= (yx)^{(+)}= (\lambda xx)^{(+)}.\]
Since $xx\in PAL$, the longest palindromic suffix $Q$ of  $\lambda xx$ has a length $|Q|\geq 2|x|$.
Thus
\[ |(\lambda xx)^{(+)}| \leq |\lambda xx \lambda^{\sim}|= 2|y|< |yxy|,\]
which is absurd. Hence, $X$ has to be a suffix code and then bifix as $X\subseteq PAL$.
\end{proof}

\begin{remark}
As a consequence of Lemma~\ref{lemma:morcon00} and
Proposition~\ref{prop:bifixpal}, every code $X$ having a morphic-conservative $\psi_{X}$ satisfies the hypotheses of Propositions~\ref{prop:pref},~\ref{prop:prefcode},~\ref{prop:prefcode1}, and~\ref{prop:p2e3}, so that all properties P1--P4 in Proposition~\ref{prop:basicp} admit suitable generalizations for $\psi_{X}$. Let us highlight in particular the following:
\end{remark}

\begin{prop}If $\psi_X$ is morphic-conservative, then it is injective.
\end{prop}
\begin{proof} From Proposition~\ref{prop:bifixpal} the code $X$ has to be bifix, so that the result follows from 
Proposition~\ref{prop:prefcode1}.
\end{proof}

Let us observe that in the preceding proposition one cannot replace morphic-conservative with conservative. Indeed, for instance, if $X=\{a, ab\}$ then $\psi_X$ is conservative (see Example \ref{ex:cons})  but it is not injective, since $\psi_X(aba)= \psi_X(abab)$.

The following theorem relates the two notions of conservative and morphic-conservative palindromization map.

\begin{thm}\label{thm:mcc} The map $\psi_X$ is morphic-conservative if and only if  $X\subseteq PAL$, $X$ is prefix, and $\psi_X$ is conservative.
\end{thm}

For the proof of the preceding theorem we need the following

\begin{lemma}\label{lemma:fipal} Let $\varphi: A^*\rightarrow B^*$ be an injective morphism and $\varphi(A)= X\subseteq PAL$. For any $w\in A^*$, $\varphi(w^{\sim})= (\varphi(w))^{\sim}$. Thus for any $w\in A^*$,  $\varphi(w)= (\varphi(w))^{\sim}$ if and only if  $w\in PAL$.
\end{lemma}
\begin{proof} The result is trivial if $w=\varepsilon$. Let us suppose $w\neq \varepsilon$ and write
$w$ as $w= a_1\cdots a_n$ with $a_i\in A$, $1\leq i \leq n$. One has  
\[ \varphi(w)= \varphi(a_1)\cdots \varphi(a_n)= x_1\cdots x_n,\]
having set $x_i= \varphi(a_i)\in X$, $1\leq i \leq n$. Since $X\subseteq PAL$, one derives
\[ (\varphi(w))^{\sim}= x_n\cdots x_1= \varphi(w^{\sim}).\]
As $\varphi$ is injective one obtains:
\[ \varphi(w)= (\varphi(w))^{\sim}= \varphi(w^{\sim}) \ \mbox{if and only if} \ \ w= w^{\sim},\]
which concludes the proof.
\end{proof}

\noindent
\begin{proof}[Proof of Theorem \ref{thm:mcc}]
$(\Rightarrow)$ Immediate from Proposition \ref{prop:bifixpal} and Lemma \ref{lemma:morcon00}.

\vspace{2 mm}

\noindent
$(\Leftarrow)$ Let $\varphi: A^*\rightarrow B^*$ be an injective morphism such that $\varphi(A)=X$ is
a prefix code and $X\subseteq PAL$. We wish to prove that for any $w\in A^*$ one has:
\[ \varphi(\psi(w)) = \psi_X(\varphi(w)).\]
The proof is by induction on the length $n$ of $w$. The result is trivial if $n=0$. If $n=1$, i.e., $w=a\in A$,
one has, as $\varphi(a)\in PAL$,
\[ \varphi(\psi(a)) = \varphi(a) = \psi_X(\varphi(a)).\]
Let us then suppose the result true up to the length $n$ and prove it for $n+1$. We can write,  by using
the induction hypothesis and the fact that $\varphi(w)\in X^*$,
$$ \psi_X(\varphi(wa)) = \psi_X(\varphi(w)\varphi(a))= (\psi_X(\varphi(w))\varphi(a))^{(+)}=
(\varphi(\psi(w))\varphi(a))^{(+)}.$$
Let $z=\psi(w)$; we need to show that $(\varphi(z)\varphi(a))^{(+)}=\varphi(\psi(wa))$.
As $\psi_X$ is conservative, by Proposition \ref{prop:cpsi} the longest palindromic suffix $Q$ of $\psi_X(\varphi(w))\varphi(a)$ $ = \varphi(z)\varphi(a)$ belongs to $X^*$. Since $\varphi(a)$ is a palindrome and $X$ is a suffix code, there exists a suffix $v$ of $z$ such that $Q=\varphi(v)\varphi(a)$. Using Lemma~\ref{lemma:fipal} one derives that $va$ is the longest palindromic suffix of $za$, so that, letting $z=uv$,
\[(\varphi(z)\varphi(a))^{(+)}=\varphi(uvau^\sim)
=\varphi((za)^{(+)})=\varphi(\psi(wa)),\]
which concludes the proof.
\end{proof}

\begin{cor}\label{cor:morfconsA} Let  $X$ be a weakly overlap-free and infix code such that  $X\subseteq PAL$. Then $\psi_X$ is morphic-conservative.
\end{cor}
\begin{proof} Trivial by Proposition \ref{prop:suff} and Theorem \ref{thm:mcc}.
\end{proof}

\begin{remark}
The hypotheses in the previous corollary that $X$ is a weakly overlap-free and infix code are not necessary in order that $\psi_X$ is morphic-conservative. For instance, let $X$ be the prefix code $X=\{aa, cbaabc\}$. 
One has that $bcX^*\cap PAL =\emptyset$. From this 
one easily verifies that for all $n\geq 0$, if $\psi_X(X^n)\subseteq X^*$, then for $x\in X$, $LPS(\psi_X(X^n)x)\subseteq X^*$.
Thus by using the same argument as in the sufficiency of  Proposition \ref{prop:cpsi}, one has that $\psi_X(X^{n+1})\subseteq X^*$. It follows that $\psi_X$  is conservative and then morphic-conservative by Theorem \ref{thm:mcc}.
\end{remark}

Let $\psi_X$ be a morphic-conservative  palindromization map and $\varphi: A^*\rightarrow B^*$ the injective morphism such that $X=\varphi(A)$ and $\varphi \circ \psi = \psi_X \circ \varphi$.
Since $X$ has to be bifix, $\varphi$ can be extended to a bijection $\varphi: A^{\omega} \rightarrow X^{\omega}$. The extension of $\psi_X$ to $X^{\omega}$ is such that for any $x\in X^{\omega}$
\[ \psi_X(x)= \varphi(\psi(\varphi^{-1}(x))).\]
For any $x\in X^{\omega}$ the word $\psi(\varphi^{-1}(x))$ is an epistandard word over $A$, so that
\[ \psi_X(X^{\omega})= \varphi(Epistand_A).\]
Therefore, one has:
\begin{prop} The infinite words generated by morphic-conservative generalized palindromization maps are images by injective morphisms of the epistandard words.
\end{prop}

Let us now consider the case when $X$ is a finite and maximal prefix code.

\begin{lemma}\label{lemma:finitemax} If $X$ is a finite and maximal prefix code over $A$ such that
$X\subseteq PAL$, then $X=A$.
\end{lemma}
\begin{proof} Let $\ell_X$ be the maximal length of the words of $X$. Since $X$ is represented by a
full $d$-ary tree, there exist $d$ distinct words $pa\in X$, with $p$ a fixed word of $A^*$, $a\in A$,
and $|pa|= \ell_X$. As for any $a\in A$, the word $pa\in PAL$, the only possibility is $p=\varepsilon$,
so that $X=A$.
\end{proof}

\begin{prop} Let $X$ be a finite and maximal prefix code over $A$. Then $\psi_X$  is morphic-conservative if and only if $X=A$.
\end{prop}
\begin{proof} The proof is an immediate consequence of Theorem \ref{thm:mcc} and Lemma \ref{lemma:finitemax}.
\end{proof}

In the case of a finite maximal prefix code the map $\psi_X$ can be non conservative. For instance,
if $X=\{a, ba, bb\}$, then $\psi_X(ba)= bab\not\in X^*$. The situation can be quite different if one refers to infinite words over $X$. Let us give the following definition.

Let $X$ be a code having a finite deciphering delay. We say that $\psi_X$ is \emph{weakly conservative}
if for any $t\in X^{\omega}$, one has $\psi_X(t)\in X^{\omega}$; in other terms the map $\psi_X: X^{\omega}\rightarrow A^{\omega}$ can be reduced to a map $\psi_X: X^{\omega}\rightarrow X^{\omega}$. In general, $\psi_X$ is not weakly conservative. For instance, if $X=\{ab, ba\}$ and 
$t\in ababX^{\omega}$, then $\psi_X(t)\not\in X^{\omega}$.

Trivially, if $\psi_X$ is conservative, then it is also weakly conservative. However, the converse is not in general true as shown by the following:

\begin{thm} If $X$ is a finite and maximal prefix code, then $\psi_X$ is weakly conservative.
\end{thm}
\begin{proof} Let $s= \psi_X(t)$ where $X$  a finite and maximal prefix code and $t\in X^{\omega}$. We recall \cite{codes} that
any maximal prefix code is right complete, i.e., for any $f\in A^*$, one has  $fA^*\cap X^* \neq \emptyset$. If $X$ is finite, then for any $f\in A^*$ and any letter $a\in A$ one has:
\[ fa^k \in X^*,\]
for a suitable integer $k$, depending on $f$ and on $a$, such that $0\leq k \leq \ell$, where $\ell=\ell_X$ is the maximal length of the words of $X$. Let $a$ be a fixed letter of $A$. We can write:
\[ s_{[n]}a^{k_n} \in X^*,\]
with $0\leq k_n\leq \ell$. Setting  $p = \lfloor \frac{n}{\ell}\rfloor$, we can write:
\[  s_{[n]}= x_1x_2\cdots x_{q_n}\lambda,\]
with $x_i\in X$, $i=1, \ldots, q_n$, $q_n\geq p$ and $|\lambda| <\ell$. Now $s_{[n]}\prec s_{[n+\ell]}$, so
that since $X$ is a prefix code, one has:
\[  s_{[n+\ell]}= x_1x_2\cdots x_{q_{n+\ell}}\lambda',\]
with $q_{n+\ell}> q_n$, $x_i\in X$, $i= q_{n+1}, \ldots, q_{n+\ell}$, and $|\lambda' |<\ell$.
Since \[\lim_{n\rightarrow\infty} x_1\cdots x_{q_n} \in X^{\omega}\] and  $\lim_{n\rightarrow\infty} x_1\cdots x_{q_n}=\lim_{n\rightarrow\infty} s_{[n]}$,  the
result follows. 
\end{proof}

\begin{cor} Let $s= \psi_X(t)$ with $t\in X^{\omega}$ be an  X-AR word. Then $s$ is the morphic image by an injective morphism of a word $w\in B^{\omega}$, where $B$ is an alphabet of the same cardinality as $X$.
\end{cor}
\begin{proof} By the preceding theorem, since $\psi_X$ is weakly conservative,  we can write:
\[s = x_1x_2\cdots x_n\cdots ,\]
with $x_i\in X$, $i\geq 1$. Let $B$ be an alphabet having the same cardinality of $X$ and $\varphi: B^*\rightarrow X^*$ be the injective morphism induced by an arbitrary bijection of $B$ and $X$. If $\varphi^{-1}$ is the inverse morphism of $\varphi$ one has:
\[\varphi^{-1}(s) = \varphi^{-1}(x_1)\varphi^{-1}(x_2)\cdots \varphi^{-1}(x_n)\cdots.\]
Setting  $\varphi^{-1}(x_i) = w_i \in B$ for all $i\geq 1$, 
one has $\varphi^{-1}(s) = w_1w_2\cdots w_n\cdots  = w\in B^{\omega}$ and $s= \varphi(w)$.
\end{proof}

Let us observe that in general the word $w\in B^{\omega}$ is not episturmian as shown by the following:

\begin{example} Let $X=\{a, ba, bb\}$ and $s= \psi_X((ababb)^{\omega})$. One has:
\[s= ababbabaababbabababbabaababbaba\cdots .\]
Let $B=\{0,1,2\}$ and $\varphi$ the morphism of $B^*$ in $X^*$ defined by the bijection $\varphi(0)= a$,
$\varphi(1)= ba$, and $\varphi(2)=bb$. One has:
\[w= \varphi^{-1}(s) =0120101201120101201\cdots,\]
and the word $w$ is not episturmian (indeed, for instance, the factor $01201$ is not rich in palindromes).
\end{example}

\section{The  pseudo-palindromization map}\label{sec:seven}

An \emph{involutory antimorphism} of $A^*$ is any antimorphism $\vt:A^*\to
A^*$ such that $\vt\circ\vt=\mathrm{id}$. The simplest example is the
\emph{reversal operator}   $R: A^*\longrightarrow A^*$ mapping each $w\in A^*$
to its reversal $ w^{\sim}$.
Any involutory antimorphism $\vt$ satisfies $\vt=\tau\circ R=R\circ\tau$ for
some morphism $\tau:A^*\to A^*$ extending an involution of $A$. Conversely, if
$\tau$ is such a morphism, then $\vt=\tau\circ R=R\circ\tau$ is an involutory
antimorphism of $A^*$.

Let $\vt$ be an involutory antimorphism of $A^*$. For any $w\in A^*$ we shall denote $\vartheta(w)$ simply by $\bar w$. We call
\emph{$\vt$-palindrome} any fixed point of $\vt$, i.e., any word $w$ such that
$w= \bar w$, and let $\PAL_\vt$ denote the set of all $\vt$-palindromes. We
observe that $\varepsilon\in\PAL_\vt$ by definition, and that $R$-palindromes
are exactly the usual palindromes.  If one makes no reference to the
antimorphism $\vt$, a $\vt$-palindrome is called a \emph{pseudo-palindrome}.

For any $w\in A^*$, $w^{\oplus_\vartheta}$, or simply $w^{\oplus}$,
denotes the shortest $\vartheta$-palindrome having $w$ as a prefix.
 If $Q$ is the
longest $\vartheta$-palindromic suffix of $w$ and $w=sQ$, then \[w^\oplus=sQ\bar
s . \]
\begin{example}
  \label{ex:oplus}
  Let $A=\{a,b,c\}$ and $\vartheta$ be defined as $\bar a=b$, $\bar c=c$. If
  $w=abacabc$, then $Q=cabc$ and $w^\oplus=abacabcbab$.
\end{example}

 We can define the \emph{$\vartheta$-palindromization map} 
$\psi_\vartheta : A^*\rightarrow  PAL_\vartheta$ by $\psi_\vartheta(\varepsilon)=\varepsilon$
and \[\psi_\vartheta(ua)=(\psi_\vartheta(u)a)^\oplus\] for $u\in A^*$ and  $a\in A$.

The following proposition extends to the case of $\vartheta$-palindromization map $\psi_{\vartheta}$ the properties of palindromization map $\psi$ of Proposition \ref{prop:basicp}
 (cf., for instance, \cite{adlADL}):
\begin{prop}\label{lemma:oplus} The  map $\psi_{\vartheta}$ over $A^*$  satisfies the following
properties: for $u,v\in A^*$
\begin{itemize}
\item[P1.]  If $u$ is  a prefix of  $v$, then $\psi_{\vartheta}(u)$ is a $\vartheta$-palindromic prefix (and suffix) of $\psi_{\vartheta}(v)$.
\item[P2.] If $p$ is a prefix of $\psi_{\vartheta}(v)$, then $p^{\oplus}$ is a prefix of $\psi_{\vartheta}(v)$.
\item[P3.] Every $\vartheta$-palindromic prefix of $\psi_{\vartheta}(v)$ is of the form $\psi_{\vartheta}(u)$ for some prefix $u$ of $v$.
\item[P4.] The map $\psi_{\vartheta}$  is  injective.
\end{itemize}
\end{prop}

The map  $\psi_\vartheta$ can be extended to infinite words as follows: let
$x=x_1x_2\cdots x_n\cdots\in A^\omega$ with $x_i\in A$ for $i\geq 1$. Since  for all $n$, $\psi_\vartheta(x_{[n]})$ is a prefix of  $\psi_\vartheta(x_{[n+1]})$,  we  can define  the infinite word $\psi_\vartheta(x)$ as:
\[\psi_\vartheta(x)=\lim_{n\rightarrow\infty}\psi_\vartheta(x_{[n]})\;.\]
 The infinite word $x$ is called the \emph{directive word} of $\psi_\vartheta(x)$,
and $s=\psi_\vartheta(x)$ the \emph{$\vartheta$-standard word} directed by $x$. If one does not make reference to the antimorphism $\vartheta$ a $\vartheta$-standard word is also called \emph{pseudostandard word}.

The class of pseudostandard  words  was introduced in \cite{adlADL}. Some interesting results about such words are also  in \cite{BdDZ1, BdDZ2}. In particular,  we mention the noteworthy  result that any pseudostandard word can be obtained, by a suitable morphism, from a standard episturmian word. 

More precisely let $\mu_{\vartheta}$ be the endomorphism of $A^*$ defined for any letter $a\in A$ as:
$\mu_{\vartheta}(a) = a^{\oplus}$, so that $\mu_{\vartheta}(a) = a$ if $ a= \bar a$ and $\mu_{\vartheta}(a) = a\bar a$, if $a\neq \bar a$. We observe that $\mu_{\vartheta}$ is injective since $\mu_{\vartheta}(A)$ is a prefix code. The following theorem, proved in \cite{adlADL}, relates the maps $\psi_{\vartheta}$ and $\psi$ through the morphism $\mu_{\vartheta}$.
\begin{thm}\label{thm:psi1}  For any $w\in A^{\infty}$, one has
$ \psi_{\vartheta}(w) = \mu_{\vartheta}(\psi(w))$.
\end{thm}
 An important consequence is that any $\vartheta$-standard word is a morphic image of an epistandard word.
 
A generalization of the pseudo-palindromization map, similar to that given in Section \ref{sec:three} for the palindromization map, is the following.
Let $\vartheta$ be an involutory antimorphism of $A^*$ and $X$ a  code over $A$. We define a map:
\[ \psi_{\vartheta, X}: X^* \rightarrow PAL_\vartheta,\]
inductively as:  $\psi_{\vartheta, X}(\varepsilon)= \varepsilon$ and for any  $w\in X^*$ and $x\in X$,
\[\psi_{\vartheta, X}(wx) = (\psi_{\vartheta, X}(w)x)^{\oplus}.\]
If $\vartheta=R$, then $ \psi_{R, X}= \psi_X$. If $X=A$ then  $\psi_{\vartheta, A} = \psi_{\vartheta}$. The map $\psi_{\vartheta, X}$ will be called the \emph{$\vartheta$-palindromization map relative to the code} $X$.
\begin{example} Let $A=\{a, b, c\}$ and $\vartheta$ be defined as $\bar a = b$ and $c = \bar c$. Let $X$ be the code $X= \{ab, ba, c\}$ and $w= ab c ba$. One has: $\psi_{\vartheta, X}(ab) = ab$, $\psi_{\vartheta, X}(abc) = abcab$ and $\psi_{\vartheta, X}(abcba) = abcabbaabcab$. 
\end{example}

 Let us now consider a  code $X$ having a finite deciphering delay. One can extend $\psi_{\vartheta, X}$ to
  $X^{\omega}$ as follows: let   $ x = x_1x_2 \cdots x_n \cdots ,$ with $x_i\in X,  i\geq 1$. For any
  $n\geq 1$, $\psi_{\vartheta,X}(x_1\cdots x_n)$ is a proper prefix of $\psi_{\vartheta, X}(x_1\cdots x_nx_{n+1})$ so that
  there exists
  			\[ \lim_{n\rightarrow \infty} \psi_{\vartheta, X}(x_1\cdots x_n) = \psi_{\vartheta,X}(x).\]
Let us observe that the word 	$\psi_{\vartheta, X} (x)$ has infinitely many $\vartheta$-palindromic prefixes. 
This implies that 	
$\psi_{\vartheta, X} (x)$ is \emph{closed under $\vartheta$}, i.e., if $w\in \Ff  \psi_{\vartheta,X}(x)$, then also $\bar w \in \Ff  \psi_{\vartheta,X}(x)$.

We remark that  the maps $\psi_{\vartheta,X}$ and their extensions  to $X^{\omega}$, when $X$ is a code with finite deciphering delay, are not in general injective. The following proposition, extending Propositions \ref{prop:prefcode} and \ref{prop:prefcode1}, can be proved in a similar way.  

\begin{prop} Let $X$ be a prefix code over $A$. Then the map $\psi_{\vartheta,X}: X^* \rightarrow PAL_\vartheta$ and its extension to $X^{\omega}$ are injective.
\end{prop}

Several concepts, such as conservative and morphic-conservative maps, and results considered in the previous sections for the map $\psi_X$ can be naturally extended to the case of the map $\psi_{\vartheta, X}$. We limit ourselves only to proving the following interesting theorem relating the maps  $\psi_{\vartheta}$ and $\psi_{\vartheta, X}$ where $X=\mu_{\vartheta}(A)$. Combining this result with
Theorem \ref{thm:psi1}  one will obtain that $\psi_{\vartheta, X}$ is morphic-conservative.

\begin{thm}\label{thm:psi2} Let $A$ be an alphabet, $\vartheta$ an involutory antimorphism, and $X= \mu_{\vartheta}(A)$. Then for any $w\in A^{\infty}$ one has:
\[ \psi_{\vartheta}(w) = \psi_{\vartheta, X}(\mu_{\vartheta}(w)).\]
\end{thm}
\begin{proof} It is sufficient to prove that the above formula is satisfied for any $w\in A^*$. The proof is obtained by making induction on the length of $w$. 

Let us first prove the base of the induction. The result is trivially true if $w=\varepsilon$. Let  $w=a\in A$. If $a=\bar a$, then $a\in X$ and $\psi_{\vartheta}(a)= a =\psi_{\vartheta, X}(\mu_{\vartheta}(a))=\psi_{\vartheta, X}(a)$. If $a\neq \bar a$, one has $\mu_{\vartheta}(a) = a\bar a \in X$ and
$\psi_{\vartheta}(a) = a\bar a= \psi_{\vartheta, X}(\mu_{\vartheta}(a))=\psi_{\vartheta, X}(a\bar a)$.

Let us now prove the induction step.  For $w\in A^*$ and $a\in A$ we can write, by using the induction hypothesis,
\begin{equation}\label{eq:apal}
 \psi_{\vartheta}(wa)= (\psi_{\vartheta}(w)a)^{\oplus}= (\psi_{\vartheta,X}(\mu_{\vartheta}(w))a)^{\oplus}.
\end{equation}
Moreover, one has:
\begin{equation}\label{eq:apal1}
 \psi_{\vartheta,X}(\mu_{\vartheta}(wa)) =  \psi_{\vartheta,X}(\mu_{\vartheta}(w)a^{\oplus}) =
(\psi_{\vartheta,X}(\mu_{\vartheta}(w))a^{\oplus})^{\oplus}=(\psi_{\vartheta}(w)a^{\oplus})^{\oplus}.
\end{equation}
We have to consider two cases. If $a=\bar a$, then $a^{\oplus}= a$, so that from the preceding formulas (\ref{eq:apal}) and (\ref{eq:apal1}) we obtain the result. 

Let us then consider the case $a \neq \bar a$. We shall prove that  $(\psi_{\vartheta}(w)a)^{\oplus}=\psi_{\vartheta}(wa)$ has the prefix  $p= \psi_{\vartheta}(w)a \bar a$, so that from property P2  of Proposition \ref{lemma:oplus} one will have $p^{\oplus}\preceq \psi_{\vartheta}(wa)$. Since $ \psi_{\vartheta} (w)a \preceq p$, one will derive that  $|\psi_{\vartheta}(wa)|=|(\psi_{\vartheta} (w)a)^{\oplus}| \leq |p^{\oplus}|$ so that $p^{\oplus}= (\psi_{\vartheta}(w)a)^{\oplus}$ from which the result will follow.
We have to consider two cases:

\vspace{2 mm}

\noindent
Case 1. $\psi_{\vartheta}(w)$ has not a $\vartheta$-palindromic suffix preceded by the letter $\bar a$. Thus
\[(\psi_{\vartheta}(w)a)^{\oplus}= \psi_{\vartheta}(w)a \bar a \psi_{\vartheta}(w),\]
so that in this case we are done.

\vspace{2 mm}

\noindent
Case 2.  $\psi_{\vartheta}(w)$ has  a $\vartheta$-palindromic suffix $u$ of maximal length preceded by the letter $\bar a$. Since $u$ is also  a $\vartheta$-palindromic prefix of $\psi_{\vartheta}(w)$, by property P3 of Proposition  \ref{lemma:oplus} there exists $v$ prefix of $w$ such that $u= \psi_{\vartheta}(v)$. Since
$\bar a u$ is a suffix of  $\psi_{\vartheta}(w)$  one has that  $ua= \psi_{\vartheta}(v)a$ is a prefix of  $\psi_{\vartheta}(w)$. By property P2 of Proposition  \ref{lemma:oplus}, $( \psi_{\vartheta}(v)a)^{\oplus}$ is a prefix of $\psi_{\vartheta}(w)$.

Since $|v|<|w|$ one has $|va|\leq |w|$. By using two times the inductive hypothesis one has:
\[ (\psi_{\vartheta}(v)a)^{\oplus}= \psi_{\vartheta}(va)= \psi_{\vartheta,X}(\mu_{\vartheta}(v)a\bar a)= (\psi_{\vartheta,X}(\mu_{\vartheta}(v))a\bar a)^{\oplus}= (\psi_{\vartheta}(v)a\bar a)^{\oplus}.\]
Hence,  $\psi_{\vartheta}(w)$ has the prefix $ua\bar a$ and the suffix $a\bar a u$, so that
$\psi_{\vartheta}(w)= \lambda a \bar a u$ with $\lambda \in A^*$ and
\[ (\psi_{\vartheta}(w)a)^{\oplus}= \lambda a\bar a u a \bar a \bar {\lambda} = \psi_{\vartheta}(w)a\bar a \bar {\lambda},\]
from which the result follows.
\end{proof}
From Theorems \ref{thm:psi1} and \ref{thm:psi2}  one derives the noteworthy:
\begin{cor}Let $A$ be an alphabet, $\vartheta$ an involutory antimorphism, and $X= \mu_{\vartheta}(A)$.  Then one has:
\[ \psi_{\vartheta} = \mu_{\vartheta}\circ \psi = \psi_{\vartheta, X}\circ \mu_{\vartheta}.\]
\end{cor}
\begin{example} Let $A=\{a, b\}$, $\vartheta$ be  defined as $\bar a = b$,  and $X=\vartheta(A) = \{ab, ba\}$. Let $w=aab$.  One has $\psi(aab) = aabaa$, $\psi_{\vartheta}(aab)= ababbaabab= \mu_{\vartheta}(aabaa)$. Moreover, $\mu_{\vartheta}(aab)= ababba$ and $\psi_{\vartheta,X}(ababba)=
ababbaabab$.
\end{example}

\small
\bibliographystyle{model1-num-names}


\begin{thebibliography}{16}
\expandafter\ifx\csname natexlab\endcsname\relax\def\natexlab#1{#1}\fi
\providecommand{\bibinfo}[2]{#2}
\ifx\xfnm\relax \def\xfnm[#1]{\unskip,\space#1}\fi
\bibitem[{de~Luca(1997)}]{deluca}
\bibinfo{author}{A.~de~Luca},
\newblock \bibinfo{title}{Sturmian words: structure, combinatorics, and their
  arithmetics},
\newblock \bibinfo{journal}{Theor. Comput. Sci.} \bibinfo{volume}{183}
  (\bibinfo{year}{1997}) \bibinfo{pages}{45--82}.
\bibitem[{Droubay et~al.(2001)Droubay, Justin, and Pirillo}]{DJP}
\bibinfo{author}{X.~Droubay}, \bibinfo{author}{J.~Justin},
  \bibinfo{author}{G.~Pirillo},
\newblock \bibinfo{title}{Episturmian words and some constructions of de {Luca}
  and {Rauzy}},
\newblock \bibinfo{journal}{Theor. Comput. Sci.} \bibinfo{volume}{255}
  (\bibinfo{year}{2001}) \bibinfo{pages}{539--553}.
\bibitem[{Arnoux and Rauzy(1991)}]{AR}
\bibinfo{author}{P.~Arnoux}, \bibinfo{author}{G.~Rauzy},
\newblock \bibinfo{title}{Repr\'esentation g\'eom\'etrique de suites de
  complexit\'e {$2n+1$}},
\newblock \bibinfo{journal}{Bull. Soc. Math. France} \bibinfo{volume}{119}
  (\bibinfo{year}{1991}) \bibinfo{pages}{199--215}.
\bibitem[{Rauzy(1985)}]{Rauzy}
\bibinfo{author}{G.~Rauzy},
\newblock \bibinfo{title}{Mots infinis en arithm\'etique},
\newblock in: \bibinfo{booktitle}{Automata on infinite words ({L}e
  {M}ont-{D}ore, 1984)}, volume \bibinfo{volume}{192} of
  \textit{\bibinfo{series}{Lecture Notes in Comput. Sci.}},
  \bibinfo{publisher}{Springer}, \bibinfo{address}{Berlin},
  \bibinfo{year}{1985}, pp. \bibinfo{pages}{165--171}.
\bibitem[{de~Luca and De~Luca(2006)}]{adlADL}
\bibinfo{author}{A.~de~Luca}, \bibinfo{author}{A.~De~Luca},
\newblock \bibinfo{title}{Pseudopalindrome closure operators in free monoids},
\newblock \bibinfo{journal}{Theor. Comput. Sci.} \bibinfo{volume}{362}
  (\bibinfo{year}{2006}) \bibinfo{pages}{282--300}.
\bibitem[{Kassel and Reutenauer(2008)}]{KREU}
\bibinfo{author}{C.~Kassel}, \bibinfo{author}{C.~Reutenauer},
\newblock \bibinfo{title}{A palindromization map for the free group},
\newblock \bibinfo{journal}{Theor. Comput. Sci.} \bibinfo{volume}{409}
  (\bibinfo{year}{2008}) \bibinfo{pages}{461--470}.
\bibitem[{de~Luca(2011)}]{adl011}
\bibinfo{author}{A.~de~Luca},
\newblock \bibinfo{title}{A palindromization map in free monoids},
\newblock \bibinfo{journal}{Proc. Steklov Institute of Mathematics} \bibinfo{volume}{274}
  (\bibinfo{year}{2011}) \bibinfo{pages}{124--135}.
\bibitem[{de~Luca(1999)}]{adl99}
\bibinfo{author}{A.~de~Luca},
\newblock \bibinfo{title}{On the combinatorics of finite words},
\newblock \bibinfo{journal}{Theor. Comput. Sci.} \bibinfo{volume}{218}
  (\bibinfo{year}{1999}) \bibinfo{pages}{13--39}.
\bibitem[{Berstel and Perrin(1985)}]{codes}
\bibinfo{author}{J.~Berstel}, \bibinfo{author}{D.~Perrin},
  \bibinfo{title}{Theory of Codes}, \bibinfo{publisher}{Academic Press},
  \bibinfo{year}{1985}.
\bibitem[{Bucci et~al.(2009)Bucci, de~Luca, and De~Luca}]{BdD}
\bibinfo{author}{M.~Bucci}, \bibinfo{author}{A.~de~Luca},
  \bibinfo{author}{A.~De~Luca},
\newblock \bibinfo{title}{Characteristic morphisms of generalized episturmian
  words},
\newblock \bibinfo{journal}{Theor. Comput. Sci.} \bibinfo{volume}{410}
  (\bibinfo{year}{2009}) \bibinfo{pages}{2840--2859}.
\bibitem[{Justin and Pirillo(2002)}]{JP}
\bibinfo{author}{J.~Justin}, \bibinfo{author}{G.~Pirillo},
\newblock \bibinfo{title}{Episturmian words and episturmian morphisms},
\newblock \bibinfo{journal}{Theor. Comput. Sci.} \bibinfo{volume}{276}
  (\bibinfo{year}{2002}) \bibinfo{pages}{281--313}.
\bibitem[{Lothaire(2002)}]{LO2}
\bibinfo{author}{M.~Lothaire}, \bibinfo{title}{Algebraic Combinatorics on
  Words}, \bibinfo{publisher}{Cambridge University Press},
  \bibinfo{year}{2002}.
\bibitem[{Lothaire(1983)}]{LO}
\bibinfo{author}{M.~Lothaire}, \bibinfo{title}{Combinatorics on Words},
  \bibinfo{publisher}{Addison-Wesley}, \bibinfo{address}{Reading MA},
  \bibinfo{year}{1983}. \bibinfo{note}{Reprinted by Cambridge University Press,
  Cambridge UK, 1997}.
\bibitem[{de~Luca and Varricchio(1999)}]{dV}
\bibinfo{author}{A.~de~Luca}, \bibinfo{author}{S.~Varricchio},
  \bibinfo{title}{Finiteness and regularity in semigroups and formal
  languages}, Monographs in Theoretical Computer Science. An EATCS Series,
  \bibinfo{publisher}{Springer-Verlag}, \bibinfo{address}{Berlin},
  \bibinfo{year}{1999}.
\bibitem[{Justin(2005)}]{J}
\bibinfo{author}{J.~Justin},
\newblock \bibinfo{title}{Episturmian morphisms and a {Galois} theorem on
  continued fractions},
\newblock \bibinfo{journal}{Theor. Inform. Appl.} \bibinfo{volume}{39}
  (\bibinfo{year}{2005}) \bibinfo{pages}{207--215}.
\bibitem[{Bucci et~al.(2008{\natexlab{a}})Bucci, de~Luca, De~Luca, and
  Zamboni}]{BdDZ1}
\bibinfo{author}{M.~Bucci}, \bibinfo{author}{A.~de~Luca},
  \bibinfo{author}{A.~De~Luca}, \bibinfo{author}{L.~Q. Zamboni},
\newblock \bibinfo{title}{On some problems related to palindrome closure},
\newblock \bibinfo{journal}{Theor. Inform. Appl.} \bibinfo{volume}{42}
  (\bibinfo{year}{2008}{\natexlab{a}}) \bibinfo{pages}{679--700}.
\bibitem[{Bucci et~al.(2008{\natexlab{b}})Bucci, de~Luca, De~Luca, and
  Zamboni}]{BdDZ2}
\bibinfo{author}{M.~Bucci}, \bibinfo{author}{A.~de~Luca},
  \bibinfo{author}{A.~De~Luca}, \bibinfo{author}{L.~Q. Zamboni},
\newblock \bibinfo{title}{On different generalizations of episturmian words},
\newblock \bibinfo{journal}{Theor. Comput. Sci.} \bibinfo{volume}{393}
  (\bibinfo{year}{2008}{\natexlab{b}}) \bibinfo{pages}{23--36}.

\end{thebibliography}

\end{document}